\documentclass[12pt]{article}
\usepackage{amsfonts}
\usepackage{color}
\usepackage{amssymb}
\usepackage{amsmath}
\usepackage{amsthm}
\usepackage{multirow}
\usepackage[unicode=true,pdfusetitle, bookmarks=true,bookmarksnumbered=false,bookmarksopen=false, breaklinks=false,pdfborder={0 0 1},backref=section,colorlinks=true, linkcolor = blue, citecolor = blue]{hyperref}

\usepackage{graphicx}

\usepackage{bbm}

\usepackage{natbib}
\usepackage{todonotes}

\usepackage{url}

\hypersetup{colorlinks=true, linkcolor=blue}

\newtheorem{theorem}{{\bf \sc Theorem}}
\newtheorem{lemma}{{\bf \sc Lemma}}
\newtheorem{example}{{\bf \sc Example}}

\newtheorem{proposition}{{\bf \sc Proposition}}

\newtheorem{definition}{{\bf \sc Definition}}

\newcommand{\mc}[1]{\ensuremath{\mathcal{#1}}}

\newcommand{\delrj}[1]{}

\newtheorem{assumption}{{\bf \sc Assumption}}

\begin{document}


\title{Incentivizing Resilience in Financial Networks\thanks{The authors thank participants at the \textit{Financial Risk and Network Theory Conference 2016}, Cambridge University, UK, and at the \textit{7th Annual Financial Market Liquidity Conference (2016)}, Budapest, Hungary, where this work was presented, and particularly P\'{e}ter Bir\'{o} for detailed feedback.}}

%
%
\author{Matt V. Leduc\thanks{IIASA, Schlossplatz 1, 2361 Laxenburg, Austria. Email: mattvleduc@gmail.com}  \and Stefan Thurner\thanks{(1) Section for Science of Complex Systems, Medical University of Vienna, Spitalgasse 23, 1090 Vienna, Austria;   (2) Santa Fe Institute, 1399 Hyde Park Road, Santa Fe, NM 87501, USA; (3) IIASA, Schlossplatz 1, 2361 Laxenburg, Austria; (4) Complexity Science Hub Vienna, Josefst\"adterstrasse 39. 1080 Vienna, Austria. Email: stefan.thurner@meduniwien.ac.at} %
}




\date{\small{International Institute for Applied Systems Analysis (IIASA) \\ Institute for Science of Complex Systems (Medical University of Vienna) \\ Santa Fe Institute \\    \  \\ May 2017\\ \small The published version of this article appeared in \textit{Journal of Economic Dynamics \& Control 82 (2017) 44-66} and is available at http://dx.doi.org/10.1016/j.jedc.2017.05.010}   }
%


%
%

\maketitle

\begin{abstract}
When banks extend loans to each other, they generate a negative externality in the form of systemic risk. They create a network of interbank exposures by which they expose other banks to potential insolvency cascades. In this paper, we show how a regulator can use information about the financial network to devise a transaction-specific tax based on a network centrality measure that captures systemic importance. Since different transactions have different impact on creating systemic risk, they are taxed differently. We call this tax a Systemic Risk Tax (SRT). We use an equilibrium concept inspired by the matching markets literature to show analytically that this SRT induces a unique equilibrium matching of lenders and borrowers that is systemic-risk efficient, i.e. it minimizes systemic risk given a certain transaction volume. On the other hand, we show that without this SRT multiple equilibrium matchings exist, which are generally inefficient. This allows the regulator to effectively  stimulate a `rewiring' of the equilibrium interbank network so as to make it more resilient to insolvency cascades, without sacrificing transaction volume. Moreover, we show that a standard financial transaction tax (e.g. a Tobin-like tax) has no impact on reshaping the equilibrium financial network because it taxes all transactions indiscriminately. A Tobin-like tax is indeed shown to have a limited effect on reducing systemic risk while it decreases transaction volume.



\end{abstract}

\thispagestyle{empty}

\setcounter{page}{0} \newpage

\section{Introduction}


Systemic risk is a property of interconnected systems, by which the failure of  an initially small set of entities can lead to the failure of a significant part of the system. 
Mechanisms by which such failures can spread have been studied in financial and economic networks, where different institutions (e.g. banks, funds, insurance companies, etc.) are exposed to each other through a network of assets and liabilities (e.g. \cite{eisenberg2001systemic}, \cite{boss2004network}, \cite{gai2010contagion} or  \cite{amini2013resilience}). The insolvency of a particular institution can then precipitate other institutions into insolvency, thus generating an insolvency cascade threatening the whole system. It is now understood that systemic risk is a network property and thus different network topologies exhibit different levels of resilience to insolvency cascades. For example, \cite{acemoglu2013systemic}, \cite{ElliotGolubJackson2014} or \cite{glasserman2015likely} study how the level of interconnectedness, in conjunction with exogenous shocks, affects the resilience of financial or economic networks.

Since systemic risk is closely related to network structure, managing systemic risk can be understood as attempting to optimally shape the architecture of the financial network. While some work has studied the formation of financial networks (e.g. \cite{farboodi2014intermediation,zawadowski2013entangled,babus2016formation,anufriev2016model}), less work has been devoted to controlling the incentives that institutions (e.g. banks) may have to form a resilient network. 
Moreover, the main policies and regulations currently employed or under consideration do not emphasize network structure. In financial systems, one such policy consists of setting capital buffers or reserve requirements, especially for systemically important institution\footnote{The Basel III framework acknowledges systemically important financial institutions (SIFI) and argues for increasing their capital requirements. See for example \cite{georg2011basel}.}. 
 Such strategies largely fail to address the essence of the problem since they neglect the financial networks aspect altogether\cite{poledna2016basel}.
Taxes imposed on banks in the form of contributions to a rescue fund have also been implemented\footnote{For example, the International Monetary Fund has proposed such a tax, the `financial stability contribution' (FSC). This means a contribution made by financial institutions to reserves for eventual crises. Such taxes have also been proposed in many countries.}. Finally, financial transaction taxes (FTT) have also been proposed\footnote{Unlike a bank tax, a financial transactions tax (FTT) is a tax placed on certain types of financial transactions. Such FTTs are being considered in various countries. A goal of such taxes is to curtail financial market volatility (see, for example, \cite{summers1989financial} or \cite{tobin1978proposal}). Related empirical research generally remains ambiguous about Granger causality between FTTs and market volatility (e.g. \cite{mcculloch2011tobin}, \cite{matheson2012security}).}.
Such FTT's however fail to capture the idea that transactions between different counter-parties may have drastically different impacts on the resilience of the whole financial system, as this depends critically on their respective positions (and thus their centrality) in the network of assets and liabilities (\cite{poledna2016basel}). For instance, a borrowing bank that borrows from a systemically important lending bank may inherit its systemic importance. Indeed if the former becomes insolvent, then it may cause the insolvency of the latter, which can then initiate a large insolvency cascade. On the other hand, if the borrowing bank borrows from a lending bank with low systemic importance, then the insolvency of the former will have little impact of the system, even if it causes the insolvency of the latter (\cite{thurner2013debtrank}, \cite{poledna2014elimination}). More and more now, information about the topology of financial or interbank networks is available to regulators (e.g. Central Banks). This  allows them to measure the impact of different transactions on the resilience of the whole system. 

In this paper, we show how a regulator can use such information about the topology of the interbank system to design incentives that help create a more resilient system. Using this information, the regulator can design a \em transaction-specific \em tax that discriminates among the possible transactions between different lending and borrowing banks. By taxing transactions between different counter-parties differently, a regulator can effectively control the architecture of the financial system. A regulator can use this \textit{systemic risk tax (SRT)} to select an optimal equilibrium set of transactions that effectively `rewire' the interbank network so as to make it more resilient to insolvency cascades. This tax was previously introduced and simulated using an agent-based model in \cite{poledna2014elimination}. We will prove  analytically that this can be done \textit{without} reducing the total credit  (transaction) volume and thus without making the system less efficient. This leads to the notion of a \textit{systemic risk-efficient equilibrium}. The intuition behind this result is that under the SRT, any given transaction volume is exchanged under a different network configuration, which creates less systemic risk. Under this desired configuration, transactions remains untaxed, whereas under other configurations, transactions are taxed according to how much they increase systemic risk. The systemic risk tax in \cite{poledna2014elimination} is based on a notion of network centrality (e.g. \cite{battiston2012debtrank}) and can be easily implemented using information about the topology of the interbank networks and the banks' capitalization.

To illustrate those facts, we study a stylized economy in which institutions (e.g. banks) are hit by liquidity shocks from their clients and then trade that liquidity with other banks. We derive a strategic equilibrium in which borrowing banks prefer to borrow from banks offering the best terms (lowest borrowing rates) while lending banks, on the other hand, manage their risk by setting a risk premium according to the probability of default of the borrowing bank. This results in the creation of an interbank network of financial exposures (loans), which carry default risk. Our equilibrium concept is inspired by the literature on matching markets (e.g. \cite{GaleShapley1962}, \cite{roth1992two})\footnote{It also bears similarities to equilibrium concepts found the literature on network formation games (e.g. \cite{jackson1996strategic}).}. Borrowing banks are thus matched to lending banks in a way that reflects their preferences for one another. We make no assumptions on the topology of the interbank system, instead allowing it to emerge endogenously from the banks' rational decisions.

We start by showing that under a standard bilateral contracting mechanism, in which the lending rate is set according to the borrower's default risk, multiple network configurations can arise in equilibrium, and most of them may present `high' systemic risk. We then show that the proposed SRT allows a regulator to select a \em unique \em equilibrium network that is efficient in the sense that it presents the lowest systemic risk given a certain transaction volume. We also show that a standard financial transaction tax (FTT), such as a Tobin tax, does not eliminate the multiplicity of equilibria and reduces transaction volume, while having only a minimal effect on decreasing systemic risk. Indeed, a standard FTT has no effect on controlling the topology of the financial system since it taxes all transactions indiscriminately. In this sense, a SRT can be understood as a generalization of a standard FTT, where each particular transaction can be taxed differently, thus allowing a regulator to select distinct equilibrium configurations.

We provide some additional results, such as that a risk management strategy by which lending banks favor the least risky borrowers leads to a unique equilibrium network that is generally inefficient, i.e. it presents higher systemic risk than can be achieved with the SRT. 

The paper is structured as follows: In Section \ref{sec:IBmarket}, we present a simple model for the formation of an interbank network in which banks extend loans to each other, thereby generating dynamically a network of assets and liabilities. Our equilibrium concept is introduced. In Section \ref{sec:SR}, we examine how this network of assets and liabilities creates systemic risk and how the systemic risk of any set of transactions can be quantified. We then introduce the systemic risk tax and compare its performance to that of a Tobin-like tax.  Section \ref{sec:conclusion} concludes. For clarity, all proofs are presented in the appendix (Section \ref{sec:appendix}), along with a detailed analysis of the properties of equilibrium interbank networks under different regimes (no tax, Tobin-like tax, systemic risk tax). We also examine different risk management strategies on the part of lending banks and some model variations.

\section{The Formation of the Financial Network}
\label{sec:IBmarket}

Before studying the impact of a particular tax policy on a firm's resilience to insolvency cascades, we must discuss how an interbank network is formed from the banks' decisions and how this creates systemic risk as an externality. In this section, we study how a financial network is formed by the matching of borrowers to lenders. This matching emerges from the strategic decisions of banks, since borrowers have preferences for lenders based on the lending rates they offer. We introduce an equilibrium concept inspired by the literature on matching (e.g. \cite{GaleShapley1962}).

\subsection{Assets, Liabilities and an Exogenous Bankruptcy Mechanism}
We study a stylized economy in which there is a set $\mathcal{N}$ of $n$ banks and a time horizon $t \in [0,T]$,
where $T$ is a random terminal time. At any time $t$, each bank $i \in \mc{N}$ owns 
an (external) risky asset of value $Y^i_t$ and an (external) long-term liability $Z$ of constant value. We assume that $Y^i_0 > Z$ so that  every bank is initially solvent.

In addition to these (external) assets and liabilities, each bank has other assets and liabilities related to the conduct of normal banking operations at each discrete time point in the set $\{1,2,...,\lfloor T \rfloor \}$. These are: (i) the interbank assets $A^{i,IB}_t$ and liabilities $L^{i,IB}_t$ resulting from interbank lending and borrowing; (ii) assets $A^{i,HH}_t$ and liabilities $L^{i,HH}_t$ resulting from loans to and deposits from client households. It also has an additional risk-free asset (e.g. a bond) $X_t^i$ in which it invests some household deposits. We can then define the equity of a bank $i$ as its assets minus its liabilities: $E_t^i = Y^i_t + A^{i,IB}_t + A^{i,HH}_t+X_t - Z - L^{i,IB}_t - L^{i,HH}_t$. The balance sheet of a bank $i$ at time $t$ is shown in Table \ref{tab:balance_sheet}.

 \begin{table}
\begin{center}
\begin{tabular}{c||c|c}
  & Assets & Liabilities \\  
  \hline \hline & & \\
            Interbank & $A^{i,IB}_t$  & $L^{i,IB}_t$ \\
                             & & \\
  \hline  & & \\
           Household-related & $A^{i,HH}_t$ & $L^{i,HH}_t$ \\
                && \\
   					& $X^i_t$ & \\
					  & & \\
               \hline & &\\
               External  & $Y^i_t$ & $Z$ \\  
                                   & & \\
  \hline \hline Equity &  & $E^i_t$
\end{tabular}
\end{center}
\caption{Bank $i$'s Balance Sheet at Time $t$. \textit{It consists of: (i) the interbank assets $A^{i,IB}_t$ and liabilities $L^{i,IB}_t$ resulting from interbank lending and borrowing; (ii) assets $A^{i,HH}_t$ and liabilities $L^{i,HH}_t$ resulting from loans to and deposits from client households; (iii) an additional risk-free asset (e.g. a bond) $X_t^i$ in which bank $i$ invests some household deposits; (iv) an external (risky) asset $Y_t^i$ and an external liability $Z$. We can then define the equity of a bank $i$ as its assets minus its liabilities: $E_t^i = A^{i,IB}_t + A^{i,HH}_t+X_t+ Y^i_t  - L^{i,IB}_t - L^{i,HH}_t - Z$. }}
\label{tab:balance_sheet}
\end{table}%

At any time $t$, a bank $i$ can be in either of two states $\theta^i_t \in \{0,1\}$, where $\theta^i_t = 1$ means that the bank is `bankrupt' whereas $\theta^i_t = 0$ means that the bank is `not bankrupt'. A bankruptcy occurs when its equity becomes negative, i.e. $E^i_t<0$. To generate exogenous defaults, 
we assume that at any time $t$, the risky asset $Y_t^i$ can undergo a negative jump and lose its full value. In the absence of this negative shock, the risky asset preserves its initial value $Y_0^i$. The value of the risky asset $Y^i_t$ can thus be expressed as $Y^i_t = Y_0^i (1- \mathbbm{1}_{\{ N_t > 0\}})$, where $N_t$ is a counting process with $N_0=0$ and hazard rate $\gamma^i$ so that $(N_s - N_t) \sim Poiss(\gamma^i \cdot (s-t))$. If this event happens for the first time at some time $t$, then $t=T$ (the terminal time). At time $t=T$, the bankrupt bank becomes unable to honor its interbank loans and thus may cause the bankruptcy of other banks. Before studying such insolvency cascades, we will first examine how the interbank assets and liabilities are formed dynamically.

\subsection{Equilibrium Matching of Lenders and Borrowers}
\subsubsection{Liquidity Shocks}
\label{sec:liquidity_shocks}

In this section, we describe the strategic interactions that drive the formation of the interbank network of assets and liabilities.

At each discrete time $t \in \{0, 1, 2, ... \lfloor T \rfloor \}$, each bank $i \in \mc{N}$ receives a liquidity shock $\epsilon_t^i$ from the following distribution  

\[ \epsilon_t^i =
  \begin{cases}
    +1       & \quad \text{with prob.} \  y/2 \\
    -1    & \quad \text{with prob.} \ y/2 \\
        0    & \quad \text{with prob.} \ 1-y \\
  \end{cases}
\]
where $y \in [0,1]$. Here, $\epsilon_t^i = 1$ means that bank $i$ is in supply of one unit\footnote{In Section \ref{sec:VariableLoanSize} of the Appendix, we discuss a model extension that allows for liquidity shocks of various sizes.} of liquidity and thus that bank $i$'s household clients have deposited one unit of cash. $\epsilon_t^i = -1$ means that bank $i$ is in demand of one unit of liquidity and thus that bank $i$'s household clients want to borrow one unit of cash. $\epsilon_t^i = 0$ means that bank $i$ did \textit{not} receive a liquidity shock (is neither in demand nor in supply of liquidity). The $\epsilon_t^i$ are assumed to be i.i.d. across banks and across time.

Banks in demand of liquidity will try to borrow from banks that are in excess of liquidity on the interbank market. At each $t$, those liquidity shocks therefore induce two sets of banks: the set of potential borrowers on the interbank market $\mc{B}_t=\{i: \epsilon_t^i < 0\}$ and the set of potential lenders on the interbank market $\mc{L}_t=\{i: \epsilon_t^i > 0\}$. 

\subsubsection{Bilateral Contracts}
\label{sec:i_rates_def_prob}

We study a simple bilateral contracting mechanism by which a bank $i\in \mc{L}_t$, in excess of liquidity, can lend money to another bank $j\in \mc{B}_t$ at a (per period) rate $r_i$, augmented by a risk premium $h_{ij}$. This risk premium reflects the banks' view of the probability that bank $j$ will default on the loan. $r_i$ is the rate that bank $i$ pays on the deposits of its household clients.

In order to simplify the analysis, we assume that the loans that will be formed between banks have identical maturities of $S$ periods. This does not affect the nature of our results and a discussion of how this assumption can be relaxed is provided in Section \ref{sec:VariableTimeToMaturity} of the Appendix.

\begin{assumption}
The time to maturity $s^t_{ij}$ of an interbank loan made at time $t$ between banks $i$ and $j$  is equal to $S \in \mathbbm{N}_+$, i.e. $s^t_{ij}= S$ periods.
\label{as:random_maturity}
\end{assumption}

A lending bank $i$ that lends to a borrowing bank $j$ thus has an expected payoff
\begin{equation}
\Pi^i_{\lambda}(j) = \frac{1}{(1+r_i)^S}(1-\rho_{t,S}^j)(1 + r_i + h_{ij})^S - 1
\label{eq:U_lender}
\end{equation}
where $\rho_{t,S}^j = \mathbbm{P}\{ {t'} \in [t,t+S]: E_{t'}^j < 0 \}$ is the probability\footnote{An expression for $\rho_{t,S}^j$ will be given later in Section \ref{sec:banks_beliefs}.} that the borrowing bank $j$ will default on this $S$-period loan and $h_{ij}$ is the risk premium charged to bank $j$ to compensate for this credit risk. Equation (\ref{eq:U_lender}) is the expected payment received from bank $j$ at maturity (assuming no recovery in the event of default) minus the amount that is lent immediately. The expected payment received at maturity is discounted at rate $r_i$, the rate at which bank $i$ borrows from its household clients (i.e., the rate paid on deposits).

Risk premia are set so as to render the lender indifferent between a risky loan and a risk-free loan. Thus, two loans have the same expected payoff. Using Eq. (\ref{eq:U_lender}), it is simple to derive a formula for a fair risk premium and this is formalized in the following lemma.
\begin{lemma}[Risk Premia]
The fair risk premium set by a lending bank $i \in \mc{L}_t$ to a borrowing bank $j \in \mc{B}_t$ with default probability $\rho_{t,S}^j$ is $h_{ij} = \frac{1 + r_i}{(1-\rho_{t,S}^j)^{1/S}}-1 - r_i$.
\label{lem:RiskPremia}
\end{lemma}
Substituting $h_{ij}$ in Eq. (\ref{eq:U_lender}), it is easy to see that the lending bank $i$ derives the same expected payoff from lending to any bank $j$. We assume throughout that default probabilities $\rho_{t,S}^j$ are common knowledge so that any lending bank can compute the risk premia $h_{ij}$. We will see later in Section \ref{sec:banks_beliefs} how $\rho_{t,S}^j$ can be computed. 

A borrowing bank $j$ that borrows from bank $i$ has an expected payoff
\begin{equation}
\Pi^j_{\beta}(i) = 1 - \frac{1}{(1+r_j)^S}(1 + r_i + h_{ij})^S. 
\label{eq:U_borrower}
\end{equation}
This is the amount that it receives immediately minus the discounted payment that it will have to make at maturity, in $S$ periods. This future payment is again discounted by the rate $r_j$ at which it can borrow money from its own household clients (the rate paid on deposits). Note that a borrowing bank does not consider its own risk of default when computing its expected payoff. Doing so would however not affect our results, as will become clear later.

A borrowing bank $j$, maximizes its expected payoff by trying to borrow from the bank that offers the lowest lending rate $r_{ij} = r_i+h_{ij}$.




\subsubsection{Preference Lists}
The bilateral contracting mechanism just described induces preferences. Indeed, banks have preferences regarding which other bank they would like to trade with: If a bank $j$ is in need of liquidity (i.e $j \in \mc{B}_t$), it would prefer borrowing from the bank that offers the lowest lending rate $r_{ij}$. All banks also have a reservation rate $\bar{r}_j$ so that they prefer not borrowing from a bank $i$ offering a rate $r_{ij}> \bar{r}_j$ that is too high. Using standard tools from the literature on matching markets\footnote{E.g. see \cite{GaleShapley1962} or \cite{roth1992two}.}, the preferences of a borrowing bank $j$ can therefore be represented by an ordered list $P^j_{\beta}$ on the set of potential lenders $\mc{L}_t$. Thus a borrower $j$'s preferences are of the form $P^j_{\beta}=a, b, j, c ...$ indicating that its first choice is to borrow from lender $a$, its second choice is to borrow from lender $b$, its third choice is \textit{not} to borrow from anyone (i.e. hence a preference for itself $j$), its fourth choice is to borrow from lender $c$ etc. We assume that the $r_i$'s can be \textit{strictly} ordered\footnote{This follows when, for example, $r_i$ is drawn from a continuous distribution. The $r_i$'s can then be (almost surely) strictly ordered and this induces strict preferences over the set of potential lenders $\mc{L}_t$.}. Lenders are then strictly ordered according to the interest rates they offer: $r_{aj} < r_{bj} < \bar{r}_j< r_{cj} < ...$ . 

On the other hand, if a bank $i \in \mc{L}_t$ is in supply of liquidity, the risk premium $h_{ij}$ renders it indifferent as to which other bank $j \in \mc{B}_t$ it lends to. We can denote its preference list by $P^i_{\lambda} = d \sim e \sim f \sim ...$ .

In the remainder of this paper, we write $P^j_{\beta}(a) \succ P^j_{\beta}(b)$ to mean that $j$ prefers borrowing from $a$ than from $b$. Similarly,  we will write $P^i_{\lambda}(d) \sim P^i_{\lambda}(e)$ to express that $i$ is indifferent between lending to $d$ or to $e$. In Section \ref{sec:other_risk_mang_strat}, we will study risk management strategies by which a lender has a strictly-ordered list of preferences over borrowers. Thus, it will then be possible to have $P^i_{\lambda}(d) \succ P^i_{\lambda}(e)$.

\subsubsection{Two-Sided Matching}

Denote by $\textbf{P}$ the set of preferences lists:
 \begin{equation}
 \textbf{P} = \{P^a_{\beta},P^b_{\beta}, P^c_{\beta}, ..., P^d_{\lambda}, P^e_{\lambda}, P^f_{\lambda}, ...\}. 
 \end{equation}
The interbank market for liquidity at time $t$ is denoted by the triple $(\mc{B}_t,\mc{L}_t,\textbf{P})$. An equilibrium outcome of the interbank market at time $t$ is a bipartite graph representing a set of matches between potential lenders and borrowers. In general, not every bank may be matched -- some banks may not be able to trade because all the liquidity may be exchanged between other banks. Some banks may also not trade because the terms of trade are such that they prefer not to trade (i.e. $\bar{r}_j < r_{ij}$).

\begin{definition}[Matching]
A matching $\mu_t$, at time $t$, is a one-to-one correspondence from the set $\mc{N}$ onto itself such that for any $b \in \mc{B}_t$, if $\mu_t(b) \neq b$, then $\mu_t(b) \in \mc{L}_t$ and for any $l \in \mc{L}_t$, if $\mu_t(l) \neq l$, then $\mu_t(l) \in \mc{B}_t$.
\end{definition}

\begin{example}
For example, let $\mc{N}=\{1,2, ..., 9\}$, $\mc{L}_t=\{1,2,3,4 \}$ and $\mc{B}_t=\{5,6,7,8,9 \}$. Then we may write 
\begin{equation*}
\mu_t=
\begin{matrix}
4 &1 & 2 & 3 & (5) \\
6 & 7 & 8 & 9 & 5 
\end{matrix}
\end{equation*}
so that $\mu_t(4)=6$ and $\mu_t(6)=4$, and thus bank $4$ lends to bank $6$, $\mu_t(1)=7$ and $\mu_t(7)=1$ and thus bank $1$ lends to bank $7$ and so on. Note that $\mu_t(5)=5$, and thus no one lends to bank $5$, which remains alone (or unmatched).
\label{ex:matching}
\end{example}

A matching induces a directed bipartite graph $M_t = \{\mc{B}_t,\mc{L}_t,\mc{E}_t\}$ on the sets of potential lenders and borrowers, where $\mc{E}_t =\{ij: \mu(i) \neq i \text{ and } i \in \mc{L}_t, \ \mu(j) \neq j \text{ and } j \in \mc{B}_t \}$ is the set of directed edges connecting lenders to borrowers. The weight of each edge is the amount of liquidity exchanged, i.e. $|\epsilon^i_t|=1$. Note that there are no self loops in $M_t$. The self match $\mu_t(5)=5$ in Example \ref{ex:matching} is thus excluded from $M_t$. This is illustrated in Fig. \ref{fig:matching}.

\begin{figure*}
  \centerline{
\includegraphics[scale=0.65]{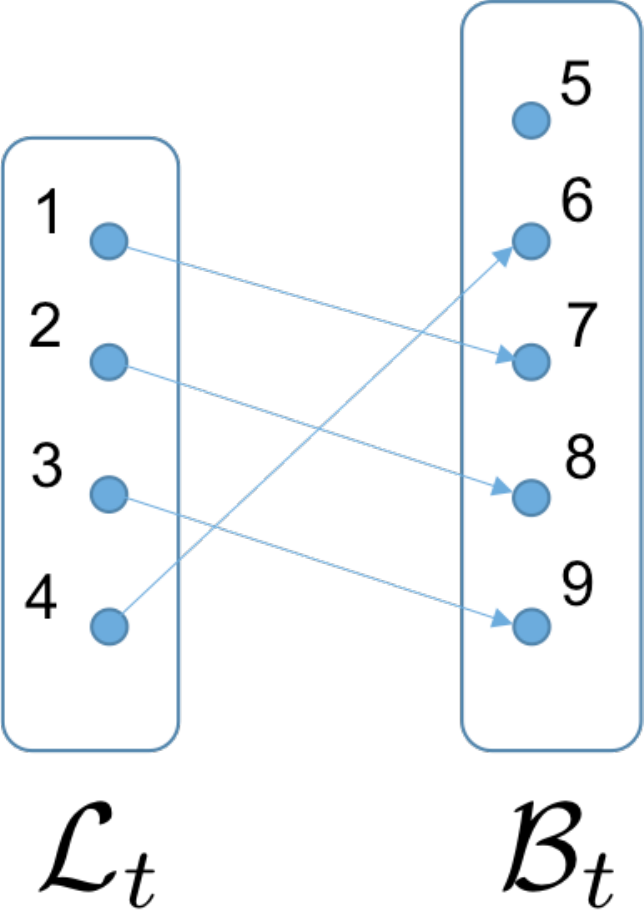}
  }
  \caption{The bipartite graph $M_t$ induced by the matching $\mu_t$ between the set of lenders $\mc{L}_t$ and the set of borrowers $\mc{B}_t$ in Example \ref{ex:matching}.}
  \label{fig:matching}
\end{figure*}

For a matching to credibly emerge from the banks' individual decisions, it has to be \em strategically stable\em, i.e. no bank should be able to gain (i.e. increase its payoff) by behaving differently. This means that no two banks on opposite sides of the market (i.e. a lender and a borrower) should be better off by dropping their matched counter-parties and trading with each other instead. Likewise, if their matched counter-parties are indifferent, no group of borrowers should benefit from agreeing to swap their assigned lenders. Finally, no single  bank should benefit from unilaterally refusing to trade with its matched counter-party. This leads us to the definition of a \textit{stable matching}, which is the basis of our equilibrium concept.

\begin{definition}[Stable Matching]
\label{def:stableMatching}
A matching $\mu^*_t$ is stable if : 

\item (I) (Pairwise deviation) For all $i,l \in \mc{L}_t$ and $k,j \in \mc{B}_t$ such that $\mu_t^*(i)=j$ and $\mu_t^*(k)=l$, it cannot be that both $P^i_{\lambda}(k) \succ P^i_{\lambda}(j)$ and $P^k_{\beta}(i) \succ P^k_{\beta}(l)$; 

\item (II) (Coalitional deviation) Let $\vec{b} \in \mc{B}_t$ be a set of borrowers such that their assigned lenders are indifferent between any borrower in $\vec{b}$, i.e. for any $j, k \in \vec{b}$, $P^{\mu^*_t(j)}_{\lambda}(j) \sim P^{\mu^*_t(k)}_{\lambda}(j)$. Then there cannot be another matching $\mu_t^{'}$ such that for all $j \in \vec{b}$, $P^{j}_{\beta}(\mu^{'}_t(j)) \succ P^{j}_{\beta}(\mu^{*}_t(j))$;

\item (III) (Unilateral deviation) For any $j \in \mc{B}_t$ such that $\mu_t^*(j) \neq j$, it cannot be that  $P^j_{\beta}(j) \succ P_{\beta}^j(\mu_t^*(j))$, and for any $j \in  \mc{B}_t$, it cannot be that $P^j_{\beta}(k) \succ P_{\beta}^j(\mu_t^*(j))$ for some $k \in  \mc{L}_t$ such that $\mu_t^*(k) = k$.

\label{def:stable_matching}
\end{definition}

Condition $(I)$ states that in a stable matching, no two banks on opposite sides of the market (i.e. a lender and a borrower) can benefit from dropping their current trading partners and agreeing to trade together instead. Condition (II) states that no group of borrowing banks can agree to swap counter-parties if their lenders are indifferent. Finally, condition $(III)$ simply states that in a stable matching, no single borrowing bank can benefit from \textit{unilaterally} breaking its current trading relationship and not trading with anyone or trading with a different unmatched lender.

This concept of a stable\footnote{Due to the presence of Condition II, some would actually call this a \emph{coalitionally stable} matching. However for brevity we will simply call it \emph{stable} throughout the article.} matching is similar to that introduced in \cite{GaleShapley1962}, which has been widely used in two-sided matching markets such as the matching of students to schools (\cite{abdulkadiroglu2003school}), medical school graduates to hospitals (\cite{roth1984evolution}, \cite{roth1999redesign}), as well as kidney donors to recipients (\cite{roth2003kidney}), for example\footnote{Our equilibrium concept is also similar to concepts introduced in the network formation games literature, e.g. \cite{JacksonWatts2002}, \cite{jackson1996strategic}. A key difference is that we form a network dynamically through a sequence of equilibrium matchings, instead of forming it statically. Our equilibrium concept is thus better adapted to the formation of a financial network, which takes place dynamically, through the formation of equilibrium matchings between lenders and borrowers.}. The standard matching problem however typically assumes that agents on each side of the market have strict\footnote{See \cite{irving1994stable} for generalizations that account for indifference through partial orders.} preferences over the other side. While condition (I) captures this, condition (II) allows us to deal with the indifference of the lenders, which occurs with the contracting mechanism introduced in Section \ref{sec:i_rates_def_prob}. To handle this case, we allow agents on the borrowing side to swap counter-parties if they benefit from doing so.

Given a stable matching $\mu_t^*$ at time $t$, define the liquidity exchanged (or transaction volume) as 

\begin{equation}
Vol(\mu_t^*) = \frac{1}{2} \sum_{i \in \mc{L}_t \bigcup \mc{B}_t} \mathbbm{1}_{ \{ \mu_t^*(i) \neq i \} }.
\end{equation}

\begin{proposition}[Equilibrium Multiplicity under Bilateral Contracting]
Let $(\mc{B}_t,\mc{L}_t, \textbf{P})$ be a market for liquidity at time $t$ and let $i \in \mc{L}_t$ and $j \in \mc{B}_t$. Under a bilateral contracting mechanism, any matching $\mu_t$ such that $r_{ij} < \bar{r}_j$ for any $\mu_t(i)=j$ and $r_{ij} < r_{kj}$ for any $k \in  \mc{L}_t$ such that $\mu_t(k)=k$ is stable, i.e. $\mu^*_t=\mu_t$. We denote by $\mc{EQ}_t$ the set of such equilibria. Moreover, the trading volume at time $t$ is bounded as follows: $Vol(\mu_t^*) \leq min(|\mc{B}_t|,|\mc{L}_t|)$.
\label{prop:EqMultiplicity}
\end{proposition}

The above proposition says that any matching such that the lending rate charged is strictly lower than a bank's reservation rate $\bar{r}_j$ and is strictly lower than the lending rate offered by any unmatched lender can be sustained in equilibrium. In fact, the lenders being indifferent as to which bank they lend to, they will agree to trade with any borrowing bank. Borrowing banks have strict preferences as to which bank they borrow from: they favor a lender that offers the lower rate and trade with it, if it is available. If the rate offered is higher than its reservation rate, the borrowing bank does not trade with that lender. The total transaction volume exchanged (the total number of loans extended) is bounded by the cardinality of the smallest of the sets of lenders and borrowers. 

As we will see in Section \ref{sec:SR}, this multiplicity of equilibria leads to many possible network configurations with varying levels of systemic risk. We first examine how an equilibrium matching $\mu^*_t$ affects the banks' balance sheets.

\subsection{Effect of Equilibrium Matching on Balance Sheets and Interbank Network}

\subsubsection{Household-Related Assets and Liabilities}
When bank $i$ receives a cash deposit from household clients (i.e. a liquidity shock $\epsilon_t^i = 1$), it tries to lend it to another bank on the interbank market. Deposits have a maturity of $S$ periods. If no other bank needs to borrow it, then bank $i$ invests it in an external\footnote{This assumption does not change the nature of our results, but is convenient. It makes the equity $E^i_t$  independent of interbank activities and thus avoids pathological cases where $E^i_t$ could turn negative for no meaningful reason (a concern in models like \cite{eisenberg2001systemic}).} risk-free asset (e.g. a bond) $X^i_t$ for the whole duration of the deposit (i.e. for $S$ periods). The sum of all household deposits made at times $t'<t$ and with any remaining time to maturity is denoted by $L^{i,HH}_t \geq 0$. 

When households wish to borrow cash from bank $i$, bank $i$ needs to obtain that cash by borrowing on the interbank market. If it succeeds in finding that money on the interbank market, it extends the loan to the households. Household loans have a maturity of $S$ periods. If bank $i$ cannot find that money on the interbank market, it simply declines to extend the loan to the households. The sum of all loans made to households at times $t'<t$ and with any remaining time to maturity is $A^{i,HH}_t \geq 0$.

\subsubsection{Interbank Assets and Liabilities}
\label{sec:IB_ass_liab}
The interbank assets of bank $i$ at time $t$ are loans extended to other banks, so that $A^{i,IB}_t = \sum_{j \neq i} A^{ij,IB}_t$, where $A^{ij,IB}_t \geq 0$ is the total amount currently loaned to bank $j$ at any time $t$. $A^{ij,IB}_t$ can be expressed as $\sum_{t':s_{ij}^{t'}>0} a^{ij,IB}_{t'}$, where $a^{ij,IB}_{t'}$ is a loan extended from $i$ to $j$ at time $t'$. Thus, $A^{ij,IB}_t$ denotes the sum of loans extended by $i$ to $j$ with some remaining time to maturity.
Likewise, the interbank liabilities at time $t$ are borrowings from other banks, i.e. $L^{i,IB}_t = \sum_{j \neq i} L^{ij,IB}_t$, where $L^{ij,IB} \geq 0$ is the total amount currently borrowed from bank $j$ at any time $t$. 
This consists of all the loans extended by $j$ to $i$ at times $t' \leq t$ and with any remaining time to maturity.
Note that by symmetry, $L^{ij,IB}_t  = A^{ji,IB}_t$ (i.e. a liability of $i$ to $j$ is an asset of $j$).

Note that in all the above, we neglect the effect of interest rates payments on balance sheets as this would obscure the analysis and would not change the nature of our results.

Note that at any time $t$, $A^{i,IB}_t +X^i_t = L^{i,HH}_t$ since what is deposited by households is either loaned to another bank or invested in the external risk-free asset $X^i_t$ for the duration of the deposit. Likewise, $A^{i,HH}_t = L^{i,IB}_t$ since what is loaned to households is always borrowed on the interbank market. Since the equity of any bank $i$ (assets minus liabilities) is $E_t^i = A^{i,IB}_t + A^{i,HH}_t+X_t+ Y^i_t  - L^{i,IB}_t - L^{i,HH}_t - Z$, it follows that it has a particularly simple form. It is simply $E^i_t =  Y^i_t - Z$. 

As stated in Assumption  \ref{as:random_maturity}, every loan has a time to maturity of $S$ periods. 
The interbank system at time $t$ is therefore the accumulation of loans formed at any time $t' \leq t$ and with any remaining time to maturity. This can be represented by a network 
\begin{equation}
G_t = \big(G_{t-1} \setminus \{ij:t' + S = t\} \big)  \bigcup M_t
\label{eq:Gt}
\end{equation}
where $G_0 = M_0$ and $M_t$ is the directed bipartite graph induced by the stable matching $\mu^*_t$ at time $t$. A loan is a bilateral agreement between $i$ and $j$ that remains in place until its maturity. A directed edge $ij$ thus remains in place until its maturity $t=t'+S$, at which point it is removed from the graph. In graph-theoretical terms, $G_t$ is a directed multigraph, i.e. a network in which several edges can exist between any nodes $i$ and $j$. These represent the different loans that have been made during previous periods and that have not reached their maturities. As discussed earlier, each directed edge has weight $1$, since this is the nominal amount of each loan. This is illustrated in Fig. \ref{fig:IB_network}.

\begin{figure*}
  \centerline{
\includegraphics[scale=0.65]{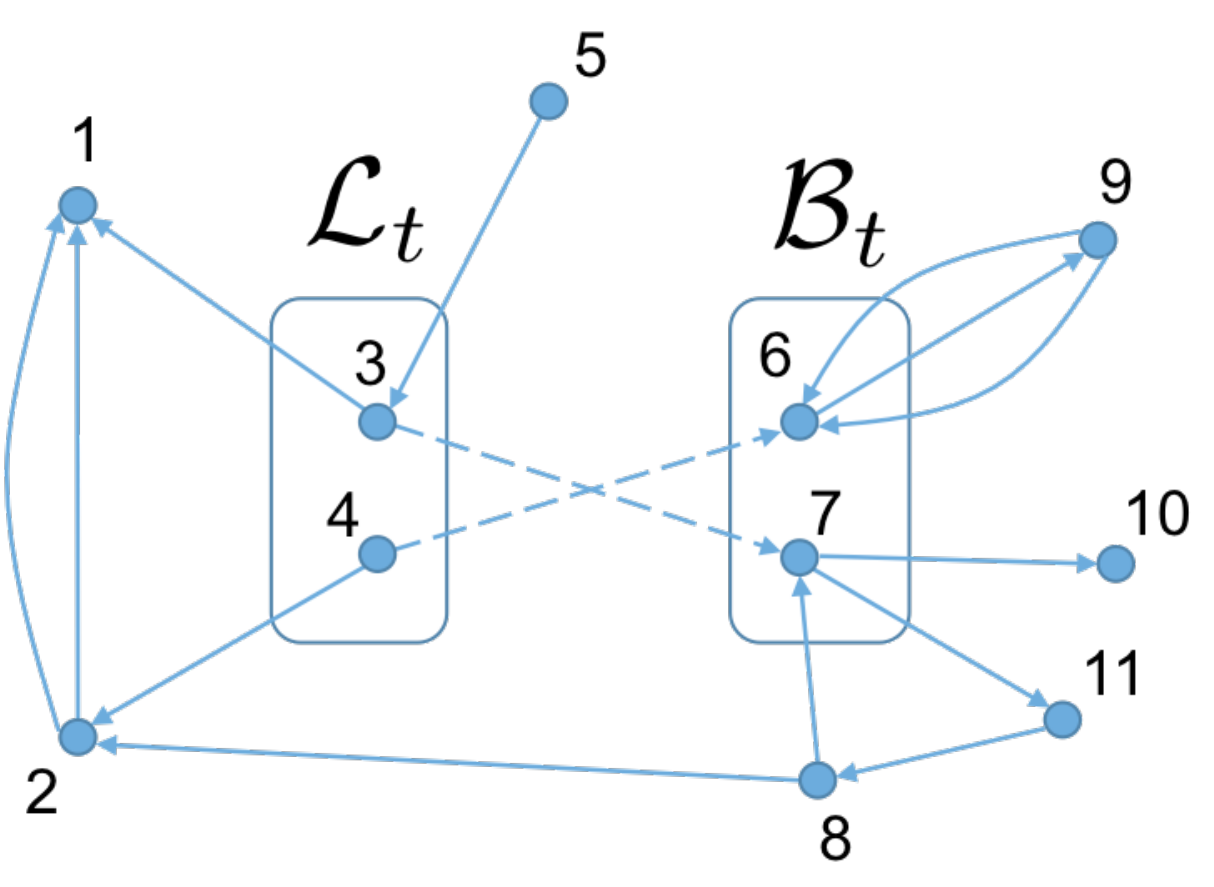}
  }
  \caption{Example of an interbank network at time $t$. \textit{In this example, only banks $3$ and $4$ are lenders at time $t$ and only banks $6$ and $7$ are borrowers at time $t$. The dotted edges represent the new loans formed at time $t$ (i.e. $M_t$ in Eq. (\ref{eq:Gt})). The solid edges represent the loans formed at previous times $t'<t$ and which have not yet reached their maturities (i.e. $t' + D> t$). The network with solid edges can be written as $G_{t-1} \setminus \{ij:t' + D= t\}$ in Eq. (\ref{eq:Gt}).} }
  \label{fig:IB_network}
\end{figure*}

We label by $\bar{A}_t$ the matrix of \textit{net} interbank exposures at time $t$. The $ij$'th entry in this matrix, $\bar{A}^{ij}_t$, represents the net exposure of bank $i$ to bank $j$, i.e. $\bar{A}^{ij}_t = A^{ij,IB}_t - A^{ji,IB}_t$. Since each exposure (loan) has value one, $\bar{A}^{ij}_t$ is simply the number of directed edges from $i$ to $j$ minus the number of directed edges from $j$ to $i$. In Eq. (\ref{eq:A_bar_matrix}) below, we show the net exposure matrix corresponding to the interbank network in Fig. \ref{fig:IB_network}.

\setcounter{MaxMatrixCols}{11}
\begin{equation}
 \bar{A}_t =
\begin{bmatrix}

 0 & -2 & -1 & 0 & 0 & 0 & 0 & 0 & 0 & 0 & 0 \\
  2 & 0 & 0 & -1 & 0 & 0 & 0 & -1 & 0 & 0 & 0 \\
  1 & 0 & 0 & 0 & -1 & 0 & 1 & 0 & 0 & 0 & 0 \\
   0 & 1 & 0 & 0 & 0 & 1 & 0 & 0 & 0 & 0 & 0 \\
    0 & 0 & 1 & 0 & 0 & 0 & 0 & 0 & 0 & 0 & 0 \\
           0 & 0 & 0 & -1 & 0 & 0 & 0 & 0 & -1 & 0 & 0 \\
       0 & 0 & -1 & 0 & 0 & 0 & 0 & -1 & 0 & 1 & 1 \\
        0 & 1 & 0 & 0 & 0 & 0 & 1 & 0 & 0 & 0 & -1 \\
         0 & 0 & 0 & 0 & 0 & 1 & 0 & 0 & 0 & 0 & 0 \\
          0 & 0 & 0 & 0 & 0 & 0 & -1 & 0 & 0 & 0 & 0 \\
           0 & 0 & 0 & 0 & 0 & 0 & -1 & 1 & 0 & 0 & 0 \\
\end{bmatrix}
\label{eq:A_bar_matrix}
\end{equation}

\subsubsection{Exogenous Default Probabilities}
\label{sec:ExDefProb}
In the previous section, we saw that a bank $j$'s equity has a particularly simple form, i.e. $E^j_t=Y_t^j - Z$. A bank can thus go bankrupt exogenously if its equity becomes negative (i.e. if $E^i_t<0$) as a result of a negative shock to the risky asset price $Y^i_t$. The exogenous default probability of a bank thus has a simple form, expressed in the following lemma.

\begin{lemma}[Exogenous Default Probability]
\label{lem:rho^i_S}
The probability that bank $j$ defaults exogenously over the next $S$ periods is
$\bar{\rho}_S^j = \big(1 - e^{-\gamma^{agg} \cdot S} \big) \frac{\gamma^i}{\gamma^{agg}}$, where $\gamma^{agg} = \sum_{j \in \mc{N}} \gamma^j$ is the sum of all hazard rates.
\end{lemma}

In the above lemma, $\bar{\rho}_S^j$ is really the probability\footnote{Note that ignore the time subscript, i.e. we write $\bar{\rho}_S^j$ instead of $\bar{\rho}_{t,S}^j$ because this exogenous default probability is constant through time. This follows from the equity having the simple form $E^j_t=Y_t^j - Z$.} that bank $j$ is the first to go bankrupt as a result of an exogenous shock to its risky asset. In the simple economy that we are studying, this first default then triggers the terminal time $T$ and thus any other bank that goes bankrupt will do so as a result of an insolvency cascade, which we will study in the next section. 


\section{Systemic Risk, Information and Incentives}
\label{sec:SR}

In this section, we examine the externalities created by the formation of the interbank market and we introduce a mechanism that allows a regulator to pin down a unique, systemic risk-efficient equilibrium.

\subsection{Quantifying Systemic Risk}
A bilateral transaction may generate a negative externality in the form of systemic risk. Indeed, when a lending bank $i$ enters into a loan agreement with a borrowing bank $j$, it not only exposes itself to the default risk of bank $j$, but it also exposes its own creditors to it. Indeed, in the event of the default of bank $j$, the loss to bank $i$'s assets may cause it to become insolvent as well (if $\bar{A}^{ij}_t > E_t^i$). This will cause bank $i$ to default on its own loans to other banks (its creditors), which can potentially cause them to become insolvent as well. This cascade of insolvencies can then propagate further throughout the network. 

Using the net exposure matrix, $\bar{A}_t$, we can compute the effect of any initial set of exogenously bankrupt banks on the system. This can be done with the following recursive dynamics, for all $i \in \mathcal{N}$. We introduce $\delta$ to denote the number of iterative steps in the cascading dynamics that takes place \em at \em the terminal time $t=T$.  

\begin{equation}
E_t^i(\delta) = max \Big(0,E^i_t(\delta-1) - \sum_{j: \sigma^j(\delta-1)=F \ \text{and} \ \bar{A}_t^{ij}>0} \bar{A}_t^{ij}  \Big)
\label{eq:E_dynamics}
\end{equation}

\[ \sigma^i(\delta) =
  \begin{cases}
    F \ (failing)      & \quad \text{if} \  E^i_t(\delta)=0 \ \text{and} \ \sigma^i(\delta-1) =H \\
    H \ (healthy)    & \quad \text{if} \ E^i_t(\delta)>0  \\
    
     I  \ (inactive)   & \quad \text{if} \ \sigma^i(\delta-1) =F \ \text{or} \ \sigma^i(\delta-1) =I  \\

  \end{cases}
\]
and
\begin{equation}
\theta_t^i(\delta) = 1 \ \text{if} \ \sigma^i(\delta) = F \text{ or } \theta_t^i(\delta - 1)=1,
\label{eq:theta_dynamics}
\end{equation}
setting the initial conditions $\sigma^i(0)=H$ and $\theta_t^i(0)=0$ for all banks $i \in \mc{N}$ and $E_t^j(0)=0$ for any $j$ part of the set of exogenously defaulted banks.
 
Since the exogenous default of any bank triggers the terminal time $T$, we may write $t=T$ in the above recursions. This recursive dynamics will stop after a finite number of steps $\bar{\delta}$ and any bank $i$ will either be in state $\theta^i_T=\theta^i_T(\bar{\delta})=1$ (bankrupt) or $\theta^i_T=\theta^i_T(\bar{\delta})=0$ (not bankrupt).
This can be seen as a reduced version of the DebtRank mechanism introduced in \cite{battiston2012debtrank} and used in \cite{thurner2013debtrank}. It is indeed closer to a standard default cascade algorithm similar to the one studied in \cite{furfine2003interbank}.

For simplicity, the above cascade mechanism assumes no recovery on defaulted loans and that a bank defaults on its loans \em only \em if it is in state `bankrupt' (i.e. $\theta^i_T(\delta) = 1$). Thus as long as a bank has positive equity $E_T^i(\delta)$, it can pay back its loans in full. This is in the spirit of \cite{eisenberg2001systemic}. These assumptions, however do not affect the nature of our results\footnote{Other variations on this insolvency cascade mechanism can be used and it does not affect the nature of our results. An alternative default mechanism is that of \cite{battiston2012debtrank}, in which a bank pays back a reduced amount to its claimants, if one of its own claimants has defaulted on a loan.}. This recursion avoids reverberations across the financial network in the sense that a bank can only transmit an insolvency shock once, i.e. when it is in state $\sigma^i(\delta) = F$ (failing). It then becomes 'inactive' (i.e., $\sigma^i(\delta+1) = I$) and no longer transmits defaults.

We can now define the impact of the bankruptcy of bank $i$ on the system.

\begin{definition}
The systemic impact of bank $i$ at time $t$ is defined as 
\begin{equation}
SI^i(\bar{A}_t,\vec{E}_t) = \sum_{j \neq i} \mathbbm{1}_{\{ \theta_t^j(\bar{\delta})=1| \theta_t^i(1)=1\}} E^j_t.
\end{equation}
\label{def:SI}
\end{definition}

$SI^i(\bar{A}_t,\vec{E}_t)$ thus represents the value of the total loss to the interbank system (as measured by the total equity lost by bankrupt banks\footnote{$SI^i$ thus ignores the equity that may be lost by non-bankrupt banks. This simple choice of systemic impact measure however does not affect the nature of our results.}) \textit{following} the bankruptcy of $i$. This quantity obviously depends on the topology of the interbank network, as shown by the dependence on $\bar{A}_t$, the matrix of net exposures. It also naturally depends on the vector of equities of all banks $\vec{E}_t = [E_t^1, E_t^2, ..., E^n_t]^{\top}$. We thus define a measure of systemic risk as follows.

\begin{definition}[Expected Systemic Loss]
\begin{equation}
\label{eq:ESL}
ESL(\bar{A}_t,\vec{E}_t) =  \sum_{j=1}^{n} \bar{\rho}_1^j \cdot SI^j(\bar{A}_{t},\vec{E}_t)
\end{equation}
\label{def:ExpectedSystemicLoss}
where $\bar{\rho}_1^j$ is the probability that bank $j$ is the first to go bankrupt exogenously over the next period.
\end{definition}
 
Equation (\ref{eq:ESL}) is the one-period-ahead expected systemic loss at time $t$, for its derivation see \cite{polednaMEXICO}. It is a convenient definition\footnote{This definition of systemic risk is that of \cite{poledna2014elimination}. It is derived on combinatorial arguments based on all possible combinations of initial defaults of institutions.} of systemic risk because it allows to separate the exogenous effects (i.e. $\bar{\rho}_1^j$) associate with external business risk from the network effects (i.e., $SI^j(\bar{A}_{t},\vec{E}_t)$) associated with contagion externalities. Using Lemma \ref{lem:rho^i_S} with $S=1$, we can express $\bar{\rho}_1^j$ as $ \big(1 - e^{-\gamma^{agg}} \big) \frac{\gamma^i}{\gamma^{agg}}$.

A loan extended from $i$ to $j$ at time $t$ (i.e. the addition of a directed edge $ij$ in the network) will thus have the following effect on systemic risk: 
\begin{equation}
\Delta ESL(ij)=ESL(\bar{A}'_{t-1} + \mathbbm{1}_{ij} - \mathbbm{1}_{ji}, \vec{E}_{t})-ESL(\bar{A}'_{t-1}, \vec{E}_{t-1})
\end{equation}
where $\mathbbm{1}_{ij}$ is a matrix of zeros with a $1$ in position $(i,j)$ and $\bar{A}'_{t-1}$ is the net exposure matrix at time $t-1$ after removing the loans that will reach their maturity at time $t$. 
This quantity can be positive or negative: certain transactions can increase systemic risk (e.g. by adding cycles of exposures in the network) while others can decrease it (e.g. by breaking cycles of exposures in the network through bilateral netting). More generally, a matching $\mu_t$ will generate a variation in the expected systemic loss as follows
\begin{equation}
\Delta ESL(\mu_t)=ESL(\bar{A}_t,\vec{E}_{t})-ESL(\bar{A}'_{t-1},\vec{E}_{t-1})
\end{equation}
where $\bar{A}_t = \bar{A}'_{t-1} + \sum_{i:i \in \mc{L}_t, i \neq \mu_t(i)} \mathbbm{1}_{\{i,\mu_t(i)\}} - \sum_{j:j \in \mc{B}_t, j \neq \mu_t(j)} \mathbbm{1}_{\{j,\mu_t(j)\}}$. In other words, $\bar{A}_t$ is the net exposure matrix formed by the  matching $\mu_t$.

\subsection{Banks' Beliefs on Total Failure Probabilities}
\label{sec:banks_beliefs}
The \textit{endogenous} default probability $q_{t,S}^j$ of borrower $j$ at time $t$ on a $S$-period loan is the probability that it defaults as a result of an insolvency cascade. It is not simple to pin down this probability. Indeed, it depends on the  evolution of the network topology over the next $S$ periods. It is not reasonable to expect banks to be able to anticipate that. A more realistic way to assess this default probability is to assume that the network topology remains fixed over the next $S$ periods. We can thus write

\begin{equation}
\label{eq:qSj}
q_{t,S}^j= \sum_{k \neq j} \mathbbm{1}\{ \theta_t^j(\bar{\delta})=1|\theta_t^k(1)=1,\bar{A}_{t-1}, \vec{E}_{t-1} \} \bar{\rho}^k_S \quad .
\end{equation}
In Eq. (\ref{eq:qSj}), we see that the indicator function  $\mathbbm{1}\{ \theta_t^j(\bar{\delta})=1|\theta_t^k(1)=1,\bar{A}_{t-1},  \vec{E}_{t-1} \} $ for the event that bank $j$ fails in a cascade caused by the exogenous failure of $k$ is conditioned on the net exposure matrix in the previous period, $\bar{A}_{t-1}$, and on the vector of equities $\vec{E}_{t-1}$.

Hence the (total) probability of failure of borrower $j$ at time $t$ on a $S$-period loan is conveniently expressed as

\begin{equation}
\rho_{t,S}^j =  \bar{\rho}^j_S + (1- \bar{\rho}^j_S)q_{t,S}^j \quad .
\end{equation}
It is the probability that a borrowing bank either defaults exogenously or as a result of an insolvency cascade. This expression conveniently separates the exogenous effects (intrinsic business risk related to the risky external assets banks invest in) from the contagion effects related to the network of interbank loans. In an ideal system free of systemic risk, a lender would thus only be concerned with the risk that a borrower fails exogenously, i.e. $\bar{\rho}^j_S$.

\subsection{Inefficiency of Equilibrium Matchings}

We will now illustrate how a bilateral contracting mechanism fails to internalize the systemic risk externality that it generates. Indeed, in a bilateral contract, a lender only considers the default risk of a borrower, while a borrower is only concerned with the interest rate that it pays. Neither party have incentives to internalize the systemic risk externality created by the transaction. 
Let us consider the example in Fig. \ref{fig:eq_multiplicity}. Here we assume that all banks fail exogenously with the same probability $\bar{\rho}_S$. Assuming $E^i_t=\$ 50$ million for all banks and each edge is a $\$ 60$-million loan, then the exogenous failure of any bank triggers the failures of all banks down its path. Then the total probabilities of failure thus follow $\rho_{t,S}^6 > \rho_{t,S}^5 > \rho_{t,S}^4$. We assume $r_3<r_2<r_1$ so that borrowers prefer lending bank $3$ over bank $2$ over bank $1$. The networks in $(a)$ and $(b)$ are the only two possible equilibria. Note that lending bank $1$ remains unmatched because it offers the highest lending rate, whereas borrowing bank $6$ remains unmatched because the borrowing rates offered by all lending banks exceed its reservation rate (i.e. $\bar{r}_6<r_{3,6}<r_{2,6}<r_{1,6}$). In Fig. \ref{fig:eq_multiplicity}(a), we see that  the transaction between lending bank $3$ and borrowing bank $5$ creates a substantial amount of systemic risk. Indeed, systemic risk \textit{spreads by lending}. To see this, note that bank $3$ has a high systemic impact: if bank $3$ defaults, it triggers the bankruptcies of banks $7$, $8$ and $9$. The loan from bank $3$ to bank $5$ then causes bank $5$ to inherit this high systemic impact. Indeed, bank $5$'s failure now triggers the bankruptcies of banks $3$, as well as banks $7$, $8$ and $9$. Bank $3$ is also part of the equilibrium matching in the other equilibrium configuration in Fig. \ref{fig:eq_multiplicity}(b), with similar consequences. 
\begin{figure*}
  \centerline{
\includegraphics[scale=0.45,angle=-90]{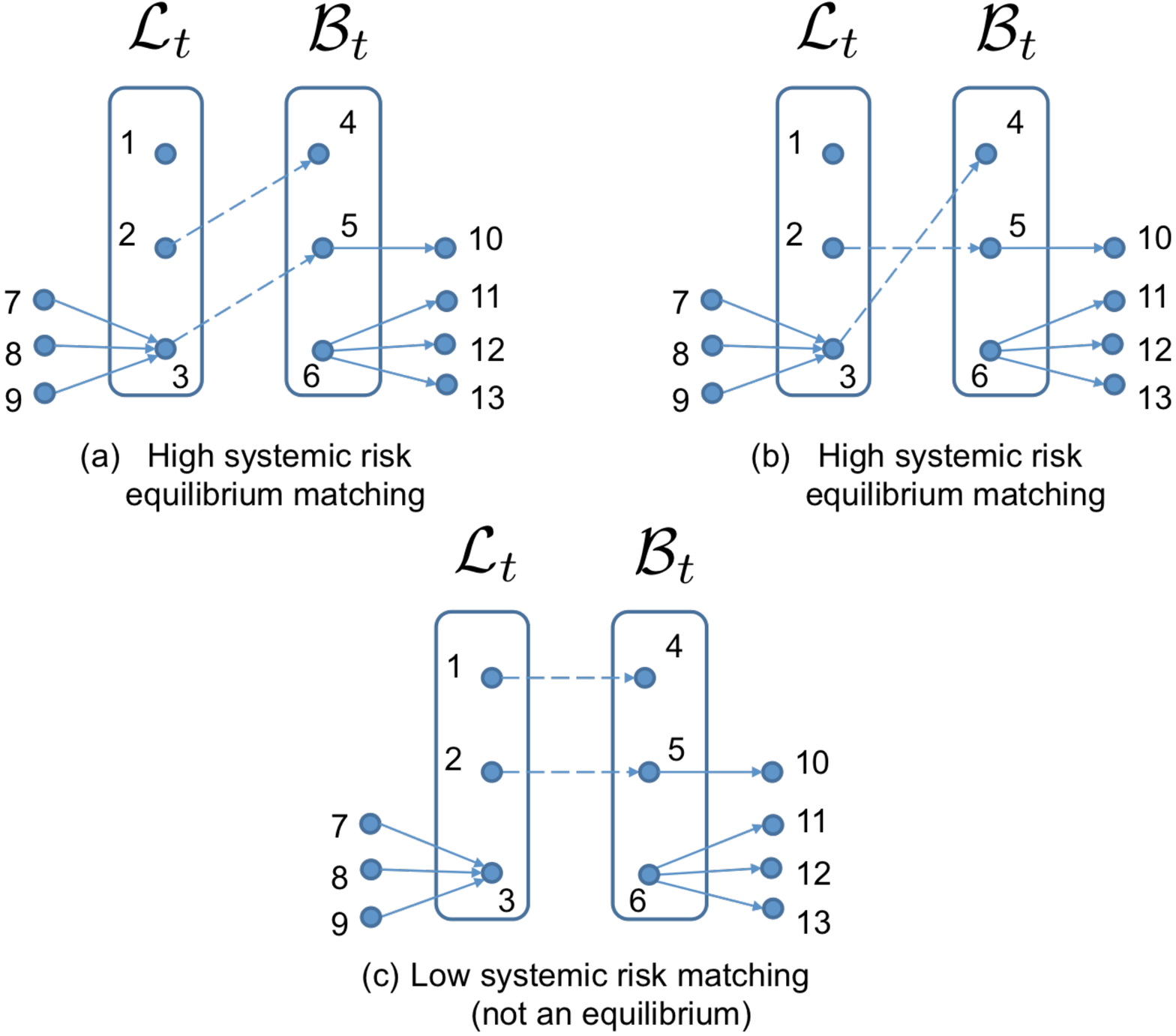}
  }
  \caption{A toy example: Equilibrium multiplicity. \textit{Assume all banks fail exogenously with the same probability $\bar{\rho}_S$. Assuming $E^i_t=\$ 50$ million for all banks and each edge is a $\$ 60$-million loan, then the exogenous failure of any bank triggers the failures of all banks down its path. Then the total probabilities of failure thus follow $\rho_{t,S}^6 > \rho_{t,S}^5 > \rho_{t,S}^4$. We assume $r_3<r_2<r_1$ so that borrowers prefer lending bank $3$ over bank $2$ over bank $1$. We also assume $r_{3,4}<r_{2,4}<r_{1,4}<\bar{r}_4$ and $r_{3,5}<r_{2,5}<r_{1,5}<\bar{r}_5$ so that borrowing banks $4$ and $5$ are willing to borrow from any lending bank while $\bar{r}_6<r_{3,6}<r_{2,6}<r_{1,6}$ so that bank $6$'s higher default risk makes borrowing too expansive. Parts (a) and (b) show the two possible equilibria. Both equilibria have high ESL. ESL in (a) is $ \bar{\rho}_1 \cdot 16 \cdot 50 $ $\$$million while in (b) it is $ \bar{\rho}_1 \cdot 13 \cdot 50 $ $\$$ million. Part (c) shows a low-ELS matching, achieving the same transaction volume, but this matching cannot be sustained in equilibrium. Its ESL is just $ \bar{\rho}_1 \cdot 10 \cdot 50 $ $\$$ million.}}
  \label{fig:eq_multiplicity}
\end{figure*}

The network configuration in Fig. \ref{fig:eq_multiplicity}(c), on the other hand, creates considerably less systemic risk. Indeed, in this equilibrium matching, the lending banks $1$ and $2$ have low systemic impact (in this example, $SI^1=0$ and $SI^2=0$). The bilateral contracting mechanism however does not allow this matching to arise in equilibrium. Indeed, bank $3$ offers the lowest rate of all lending banks and the risk premia it offers are not high enough to deter banks $4$ and $5$ from borrowing from it. The risk premium indeed only considers the borrowing banks' default probability and \textit{not} the systemic risk created by a transaction with a high systemic impact bank. Lending bank $3$ is thus necessarily part of any equilibrium matching.

The multiple equilibria that may emerge under a bilateral contracting mechanism (cf. Proposition \ref{prop:EqMultiplicity}) have different effects on systemic risk. Given a certain transaction volume, the ideal matching $\mu_t$ is the one that minimizes systemic risk $ESL(\bar{A}_t,\vec{E}_t)$. To help us characterize the different equilibria that may emerge, the following definition will be useful.

\begin{definition}[Systemic Risk-Efficient Equilibrium]
Suppose $\bar{A}_{t-1}$ is a net exposure matrix at time $t-1$. Given a market for liquidity $(\mc{B}_t,\mc{L}_t, \textbf{P})$ at time $t$, let $\bar{A}^*_{t}$ be the net exposure matrix formed by the equilibrium matching $\mu^{*}_t$. For any trading volume $v$, an equilibrium $\mu^{*,eff}_t$ is systemic risk-efficient if 
\begin{equation}
\mu^{*,eff}_t \in \underset{\{\mu_t \in \overline{\mc{EQ}}_t : \ Vol(\mu_t) = \nu \}}{\text{argmin}} ESL(\bar{A}_{t},\vec{E}_t)
\end{equation}
where $\overline{\mc{EQ}}_t$ denotes the set of matchings $\mu_t$ such that $r_{ij} < \bar{r}_j$ for all $\mu_t(i)=j$, for $i \in \mc{L}_t$ and $j \in \mc{B}_t$.
\end{definition}

Thus an equilibrium matching is systemic risk-efficient if it minimizes systemic risk, given a certain transaction volume. Note that a systemic risk-efficient matching may not always be part of the set of possible equilibria $\mc{EQ}_t$. In such a case, an equilibrium may be inefficient. This was the case in Fig. \ref{fig:eq_multiplicity}. The two possible equilibria in Fig. \ref{fig:eq_multiplicity}(a)-(b) are inefficient. Indeed, the matching in Fig. \ref{fig:eq_multiplicity}(c) is systemic risk-efficient for a transaction volume $\nu=2$, but it cannot be sustained in equilibrium.  However, in cases where $\mc{EQ}_t = \overline{\mc{EQ}}_t$, then a systemic risk-efficient matching may arise in equilibrium.   

In the next section, we introduce a tax mechanism that allows a prudent regulator to select a \em unique \em systemic risk-efficient equilibrium.

\subsection{Systemic Risk Tax (SRT)}

\subsubsection{Definition and Theoretical Results}

At any (discrete) decision time $t \in \{0,1,2,...\}$, a regulator (e.g. a Central Bank) possessing information about the current credit topology of the interbank system (i.e. knowing $\bar{A}_t$) would like to control the formation of the interbank network by influencing the matching between the sets of potential lenders $\mathcal{L}_t$ and borrowers $\mathcal{B}_t$ so as to achieve a desired level of systemic risk.
The question is thus how can she incentivize the banks so that they form a desired equilibrium matching $\hat{\mu}^*_t$? She can do this by means of a \em transation-specific \em tax, which will have the effect of reordering the borrowers' preferences for the lenders.

Let $\mc{T} = \{\tau_{ij}\}$, where $i \in \mc{L}_t$ and $j \in \mc{B}_t$. $\mc{T}$ is thus a $|\mc{L}_t| \times |\mc{B}_t|$ matrix of transaction-specific taxes. We assume $\tau_{ij} \geq 0$. $\tau_{ij} $ is the mark-up that is applied to the interest rate paid by bank $j$ when it borrows from bank $i$. The borrowing bank $j$ then pays $ r^{\mc{T}}_{ij} = r_i + h_{ij} + \tau_{ij}$ instead of just $ r_{ij}=r_i + h_{ij}$. Under $\mc{T}$, a borrower's expected payoff (cf. Eq. (\ref{eq:U_borrower})) becomes
\begin{equation}
\Pi^j_{\beta,\mc{T}}(i) = 1 - \frac{1}{(1+r_j)^S}(1 + r_i + h_{ij} + \tau_{ij})^S. 
\label{eq:U_borrower_SRT}
\end{equation}
A lender's expected payoff (cf. Eq. (\ref{eq:U_lender})) is left unchanged as the tax $\tau_{ij}$ is collected by the regulator. Thus $\mc{T}$ effectively re-orders the preferences of each borrower\footnote{Note also that since a lender is indifferent to who it lends to, or if it lends at all, the actions of the regulator (i.e. the tax $\tau_{ij}$) do not affect the lenders' equilibrium behavior. The tax only affects the \textit{borrowers'} equilibrium behavior.} over the set of lenders. This allows a regulator to create heterogeneous preferences, i.e. each borrower can now have a different preference list $P_{\beta}^j$. 

Note that since all information about the system is common knowledge, the borrowers' default probabilities $\rho_{t,S}^j$ and the lenders' baseline lending rates $r_i$, for all $i,j\in\mc{N}$, are known to the regulator as well. The latter can compute the risk premia $h_{ij}$. We also assume that the reservation rates $\bar{r}_j$, for all $j \in \mc{N}$, are known to the regulator. The banks' payoffs are thus known to the regulator. We will show that by properly choosing $\mc{T}$, the regulator can reorder each borrower $j$'s preference list $P_{\beta}^j$ such that \em any \em desired matching $\hat{\mu}_t \in \overline{\mc{EQ}}_t$ is sustained as the \em unique \em equilibrium. Since this tax allows her to pin down a systemic risk-efficient equilibrium, we will call this tax a 
\em systemic risk tax \em (SRT).

\begin{theorem}[Equilibrium Uniqueness under Systemic Risk Tax]
Let $(\mc{L}_t,\mc{B}_t,\textbf{P})$ be any market for liquidity at time $t$ and let $i \in \mc{L}_t$ and $j \in \mc{B}_t$. For any possible matching $\mu_t$ such that $r_{ij} < \bar{r}_j$ for all $\mu_t(i)=j$, there exists a SRT $\mc{T}$ such that $\mu_t^{*,\mc{T}} = \mu_t$ is the unique stable matching. The set of possible stable matchings that can be sustained as a unique equilibrium is a superset of the set of possible stable matchings that can arise without the SRT, i.e. $\overline{\mc{EQ}}_t \supseteq \mc{EQ}_t$.
\label{th:SRT_unique_match}
\end{theorem}

Theorem \ref{th:SRT_unique_match} states that an appropriate choice of SRT $\mc{T}$ can select \em any \em of the multiple equilibria that can arise under a bilateral contracting mechanism. It can also select certain matchings that could not be sustained in equilibrium without the tax. Moreover, under the SRT, this equilibrium is \em unique\em. The intuition is that the preferences of the borrowers can be arbitrarily reshuffled and this is sufficient to create any desired stable matching, irrespectively of the preferences of the lenders. Under this unique equilibrium selected by the tax, there is no coalition of banks that can agree to reshuffle their matched partners so that they all benefit from doing so. Indeed, under the tax, each borrower chooses to trade with its preferred counter-party. Even if we assume limited communication between the borrowing banks, this unique equilibrium can credibly arise from a system in which borrowing banks solicit the regulator (e.g. Central Bank) for quotes on the lending banks. These quotes are the rates $r^{\mc{T}}_{ij}$ at which they can borrow from each lender. Under the SRT $\mc{T}$, they will choose to borrow from the bank offering the lowest rate $r^{\mc{T}}_{ij}$ and this corresponds to the unique equilibrium outcome. Note also that the set of equilibria that can be uniquely sustained under the tax is \textit{larger} than the original set.

Note that in the special case when $\mc{T}_{ij}=\kappa$, for all $i,j$, then $\mc{T}$ reduces to a Tobin-like tax. A Tobin-like tax, on the other hand does not allow a regulator to induce equilibrium uniqueness, nor to increase the set of candidate matchings that can be sustained in equilibrium. It merely reduces the set of possible high-transaction volume equilibria. This is formalized in the next proposition.

\begin{proposition}[Tobin-like tax]
Let $(\mc{B}_t,\mc{L}_t, \textbf{P})$ be a market for liquidity at time $t$ and let $i \in \mc{L}_t$ and $j \in \mc{B}_t$. Let $\kappa$ be a Tobin-like tax, i.e. $r_{ij}^{\kappa} = r_i + h_{ij} + \kappa$. Then, under a bilateral contracting mechanism: 
\item (i) Any matching $\mu_t$ such that $r_{ij}^{\kappa} < \bar{r}_j$ for any $\mu_t(i)=j$ and $r^{\kappa}_{ij} < r^{\kappa}_{mj}$ for any $m \in  \mc{L}_t$ such that $\mu_t(m)=m$ is stable, i.e. $\mu^{*,\kappa}_t=\mu_t$. We denote by $\mc{EQ}_t^{\kappa} $ the set of such equilibria; 

\item (ii) The trading volume at time $t$ is bounded as follows: 
\begin{equation*}
\underset{\mu_t^{*,\kappa} \in \mc{EQ}_t^{\kappa}}{\text{max}} Vol(\mu_t^{*,\kappa}) \leq \underset{\mu_t^{*} \in \mc{EQ}_t}{\text{max}} Vol(\mu_t^{*}).
\end{equation*}
\label{prop:Tobin_tax}
\end{proposition}

A Tobin-like tax does not allow a regulator to re-order the preference lists of the borrowers. Indeed all borrowers have homogenous preferences, but some transactions become too expensive and thus not sustainable in equilibrium, hence the possible reduction in transaction volume. A Tobin tax therefore cannot allow a regulator to pin down a unique systemic risk-efficient equilibrium. An appropriate choice of the SRT $\mc{T}$ however can achieve this.

The next result has important implications. It states that a regulator can always choose a transaction-specific tax $\mc{T}$ so as to achieve lower systemic risk \em without sacrificing \em transaction volume. The intuition behind this result is that the transaction-specific tax has the effect of reordering each borrower's preferences for lenders and this is sufficient to achieve effectively any possible matching between lenders and borrowers. This can lead to matchings with lower systemic risk and/or higher trading volume. A Tobin tax, on the other hand, indiscriminately taxes every transaction equally. This has the effect of reducing the set of lenders with which a borrower is willing to trade, \textit{without} reordering the preferences of these borrowers. It therefore simply reduces transaction volume by reducing the set of possible matchings. 

\begin{proposition}[Systemic Risk under Systemic Risk Tax]
Let $(\mc{B}_t,\mc{L}_t, \textbf{P})$ be a market for liquidity at time $t$. Given a net exposure matrix $\bar{A}_{t-1}$ at time $t-1$, let $\bar{A}^{*,\mc{T}}_{t}$, $\bar{A}^{*,\kappa}_{t}$ and $\bar{A}^*_{t}$ be the net exposure matrices formed at time $t$ with a SRT $\mc{T}$, with a Tobin-like tax $\kappa$ and without tax by the equilibrium matchings $\mu^{*,\mc{T}}_t$, $\mu^{*,\kappa}_t$ and $\mu^{*}_t$, respectively. Then,

\item (i) for any $\mu^*_t \in \mc{EQ}_t$, such that $Vol(\mu^*_t) = \nu$, there exists $\mc{T}$ such that $ESL(\bar{A}^{*,\mc{T}}_{t},\vec{E}_{t}) \leq ESL(\bar{A}^*_{t},\vec{E}_{t})$ and $Vol(\mu^{*,\mc{T}}_t) \geq Vol(\mu^*_t)$; In particular, there exists $\mc{T}$  such that $\mu^{*,\mc{T}}_t$ is systemic risk efficient.
 
\item (ii) for any $\mu^{*,\kappa}_t \in \mc{EQ}_t^{\kappa}$, such that $Vol(\mu^{*,\kappa}_t) = \nu$, there exists $\mc{T}$ such that $ESL(\bar{A}^{*,\mc{T}}_{t},\vec{E}_{t}) \leq ESL(\bar{A}^{*,k}_{t},\vec{E}_{t})$ and $Vol(\mu^{*,\mc{T}}_t) \geq Vol(\mu^{*,k}_t) $.

\label{th:ESL}
\end{proposition}

Thus for any outcome of a market for liquidity under a bilateral contracting mechanism, we can design a SRT that achieves lower systemic risk and potentially higher transaction volume. The intuition is that the set of possible high-volume equilibrium matchings that an SRT can uniquely sustain is greater than the set of equilibria under a Tobin tax or no tax at all. 

Proposition \ref{th:ESL} implies that a SRT can pin down a systemic risk-efficient equilibrium for any given transaction volume. 
While a Tobin tax sacrifices transaction volume without having an optimal impact on network topology, SRT allows to minimize systemic risk given a desired transaction volume.

Fig. \ref{fig:eq_multiplicity_taxes} provides a simple illustration of Theorem \ref{th:SRT_unique_match} and Propositions \ref{prop:Tobin_tax} and \ref{th:ESL}. We make the same assumptions on model parameters as in Fig. \ref{fig:eq_multiplicity}: Parts (a) and (b) show the two possible equilibria without tax (as in Fig. \ref{fig:eq_multiplicity}a-b). Parts (c) and (d) show the equilibria under a Tobin-like tax $\kappa$ that causes $r^{\kappa}_{15} > r^{\kappa}_{25} >\bar{r}_5$ and $r^{\kappa}_{35}<\bar{r}_5$ so that the transaction between lending bank $2$ and borrowing bank $5$ becomes too expensive, since the rate charged now exceeds bank $5$'s reservation rate. The Tobin-like tax $\kappa$ however leaves $r^{\kappa}_{24}  <\bar{r}_4$ due to bank $4$ lower default probability. A Tobin tax thus leaves the first equilibrium unchanged ((c) vs (a)), while it reduces the volume of the second equilibrium ((d) vs (b)). Part (e), on the other hand, shows the unique equilibrium matching that can be achieved with a proper choice of SRT $\mc{T}$. This unique equilibrium is systemic risk-efficient for a transaction volume of two loans. One simple choice of $\mc{T}$ is simply to set $\tau_{14}=0$, $\tau_{15}=\tau_{16} >>0$, while $\tau_{25}=0$, $\tau_{24}=\tau_{26} >>0$ and $\tau_{36} >>0$. The desired matches are left untaxed, whereas the undesired ones are taxed (in this simple example, arbitrarily). This guarantees that the desired lenders are on top of each borrower's preference lists and allows this systemic risk-efficient matching to be sustained as a unique equilibrium, without reduction in volume.

\begin{figure*}
  \centerline{
\includegraphics[scale=0.45]{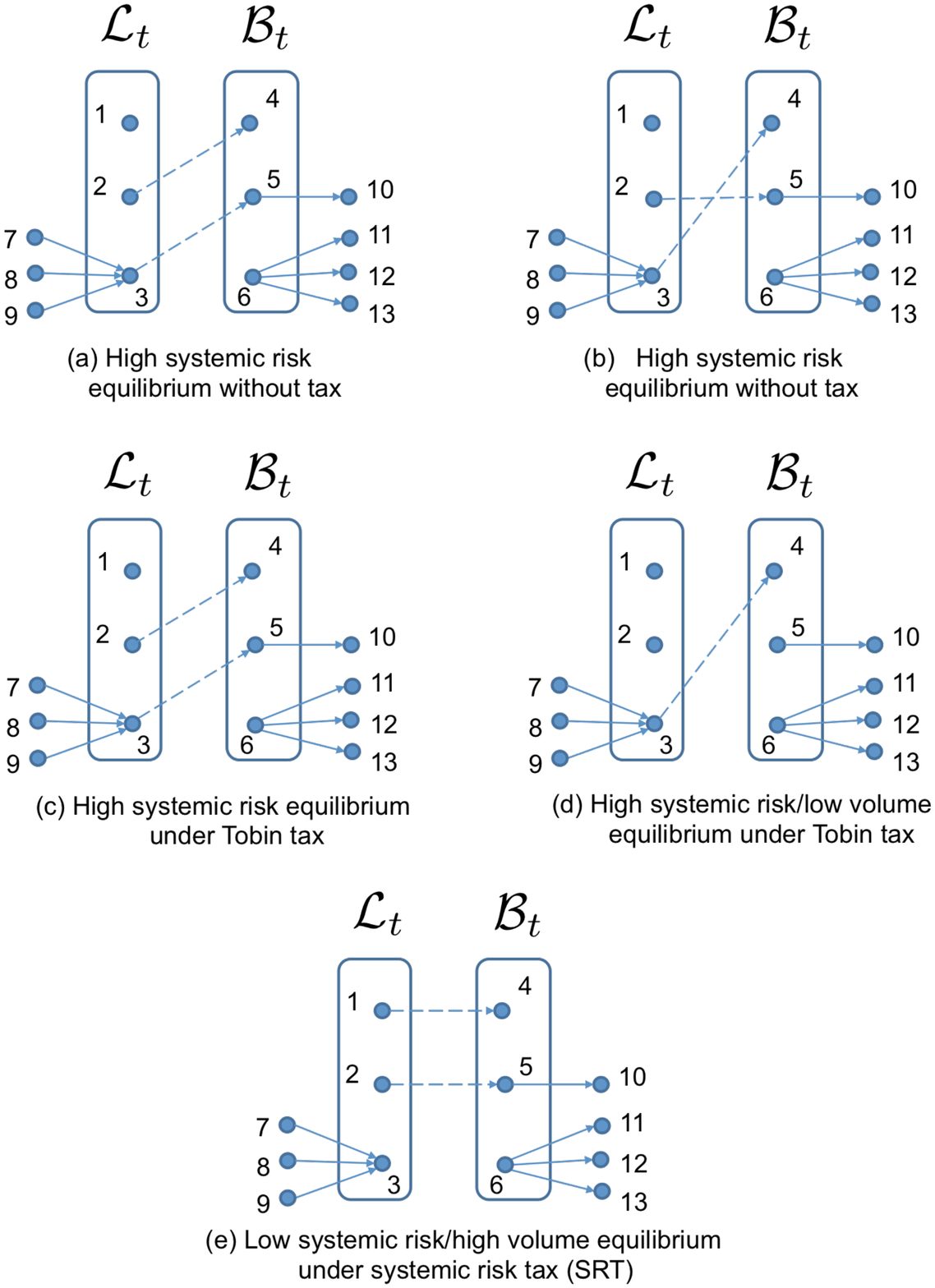}
  }
  \caption{A toy example: Systemic risk tax (SRT) leads to systemic risk-efficient equilibrium. \textit{We make the same assumptions on model parameters as in Fig. \ref{fig:eq_multiplicity}: 
  Parts (a) and (b) show the two possible equilibria without tax (as in Fig. \ref{fig:eq_multiplicity}a-b). Parts (c) and (d) show the equilibria under a Tobin-like tax $\kappa$ that causes $r^{\kappa}_{15} > r^{\kappa}_{25} >\bar{r}_5$ and $r^{\kappa}_{35}<\bar{r}_5$ so that the transaction between lending bank $2$ and borrowing bank $5$ becomes too expensive, since the rate charged now exceeds bank $5$'s reservation rate. The Tobin-like tax $\kappa$ however leaves $r^{\kappa}_{24}  <\bar{r}_4$ due to bank $4$ lower default probability. A Tobin tax thus leaves the first equilibrium unchanged ((c) vs (a)), while it reduces the volume of the second equilibrium ((d) vs (b)). Part (e), on the other hand, shows the unique equilibrium matching that can be achieved with a proper choice of SRT $\mc{T}$. This unique equilibrium is systemic risk-efficient for a transaction volume of two loans. One simple choice of $\mc{T}$ is simply to set $\tau_{14}=0$, $\tau_{15}=\tau_{16} >>0$, while $\tau_{25}=0$, $\tau_{24}=\tau_{26} >>0$ and $\tau_{36} >>0$. The desired matches are left untaxed, whereas the undesired ones are taxed (in this simple example, arbitrarily). This guarantees that the desired lenders are on top of each borrower's preference lists and allows this systemic risk-efficient matching to be sustained as a unique equilibrium, without reduction in volume.}}
  \label{fig:eq_multiplicity_taxes}
\end{figure*}

One way to make use of Proposition \ref{th:ESL} in an optimization problem is to minimize systemic risk given a target level of transaction volume. This is done in the next section.

\subsubsection{Numerical Investigation}
\label{sec:Reg_opt_prob}

As a transaction-specific tax, the SRT allows a regulator (e.g. Central Bank) to minimize systemic risk, while achieving a certain transaction volume. Suppose that she wishes to achieve transaction volume $\nu$. Then at time $t$, she can set the SRT $\hat{\mc{T}}$ by solving the following one-period-ahead optimization problem.


\begin{equation}
\hat{\mc{T}} \in \underset{\mc{T}:Vol(\mu_t^{*,\mc{T}}) = \nu}{\text{argmin}} ESL(\bar{A}_t^{*,\mc{T}},\vec{E}_t) 
\label{eq:minProb}
\end{equation}
where 
\begin{equation}
\bar{A}_t^{*,\mc{T}} = \bar{A}'_{t-1}  + \sum_{i : i \in \mc{L}_t, i \neq \mu_t^{*,\mc{T}}(i) }\mathbbm{1}_{\{i,\mu^{*,\mc{T}}_t(i)\}} - \sum_{j : j \in \mc{B}_t, j \neq \mu_t^{*,\mc{T}}(j)} \mathbbm{1}_{\{j,\mu^{*,\mc{T}}_t(j)\}}
\label{eq:minProbConstr}
\end{equation}
is the net exposure matrix formed with the equilibrium matching $\mu^{*,\mc{T}}_t$ at time $t$, and $ESL(\cdot)$ is the one-period-ahead expected systemic loss at time $t$, as defined in Definition \ref{def:ExpectedSystemicLoss}. $\bar{A}'_{t-1}$ is the net exposure matrix at time $t-1$ after removing the loans that will reach their maturity at time $t$. 



The regulator will thus choose $\hat{\mc{T}}$ such that a desired systemic risk-efficient matching $\hat{\mu}_t$ will be sustained in equilibrium, i.e. $\mu_t^{*,\hat{\mc{T}}}=\hat{\mu}_t$. Note that there can be many $\mc{T}$ yielding the same $\mu_t^{*,\mc{T}} = \hat{\mu}_t$. An economically meaningful way to design this SRT is to tax any deviation from the desired equilibrium matching $\hat{\mu}_t$ proportionally to the amount of systemic risk that it generates\footnote{This was studied with an agent-based model in \cite{poledna2014elimination}.}. The desired equilibrium itself remains untaxed.  Thus $\forall j \in \mc{B}_t$, set $\mc{T}_{\hat{\mu}_t(j),j}=0$ and set 
\begin{equation}
\mc{T}_{ij} = r_{\hat{\mu}_t(j),j} - r_{ij} + \epsilon + \zeta max(0,\Delta ESL(ij))
\end{equation}
where $\epsilon > 0$ and $\zeta$ is some scaling parameter. This has the effect of re-ordering a borrower's preferences in decreasing order of their contribution to systemic risk. Thus $r_{ij}^{\mc{T}}$ now reorders the preferences of the borrowers with the desired match on top and taxes the other matches proportionally to the risk they create.

The minimization problem in Eqs. (\ref{eq:minProb})-(\ref{eq:minProbConstr}) has a solution\footnote{Note however, that for large systems, this problem may pose computational difficulties. In such cases, approximations schemes could be used to find an approximate solution. Such a scheme consists in separating $\mc{B}_t$ into smaller subsets and then sequentially matching those subsets to $\mc{L}_t$. This is beyond the scope of this article.}, since it is just a combinatorial optimization problem over a finite set of possible matchings between the finite sets $\mc{L}_t$ and $\mc{B}_t$. To provide an illustration, we now solve this problem on a dynamically evolving complex network. 
Figure \ref{Fig1} shows the evolution of the expected systemic loss and the cumulative transaction volume in $3$ different scenarios: (i) without tax; (ii) with a Tobin-like tax; and (iii) with the SRT. In the latter case, the regulator finds the SRT $\hat{\mc{T}}$ by solving the optimization problem in Eqs. (\ref{eq:minProb})-(\ref{eq:minProbConstr}) at each time $t$. The interbank system in this example is composed of $|\mc{N}|=10$ banks.

In the upper panel, we see that a Tobin-like tax (blue curve) only has a limited effect on reducing the expected systemic loss. Moreover, in the lower panel we see that this comes at the cost of a reduction in trading volume, i.e. the number of loans extended. The SRT (green curve), on the other hand, allows to reconfigure the network of exposures in a systemic risk efficient way. It therefore does not reduce transaction volume while drastically reducing the expected systemic loss. Here the constraint is set to $\nu=Vol(\mu^*_t)$, i.e. the regulator finds the systemic risk-efficient equilibrium matching $\mu^{*,\hat{\mc{T}}}_t$ that achieves the same transaction volume as the untaxed, systemic risk-inefficient equilibrium matching $\mu^*_t$. A Tobin-like tax does not enable the regulator to achieve this. The Tobin-like tax $\kappa$ simply prevents some transactions from taking place by making them too expensive for certain borrowers (i.e. $r_{ij}^{\kappa}>\bar{r}_j$ for certain $j \in \mc{B}_t$).

With a Tobin-like tax, the reduction in systemic risk is due to the reduction in the number of loan exposures. On the other hand, with a SRT it is due to a more efficient allocation of those loan exposures across counter-parties. This means that the SRT incentivizes a specific matching of counter-parties that minimizes the expected systemic loss.

\begin{figure*}
  \centerline{
\includegraphics[scale=0.65]{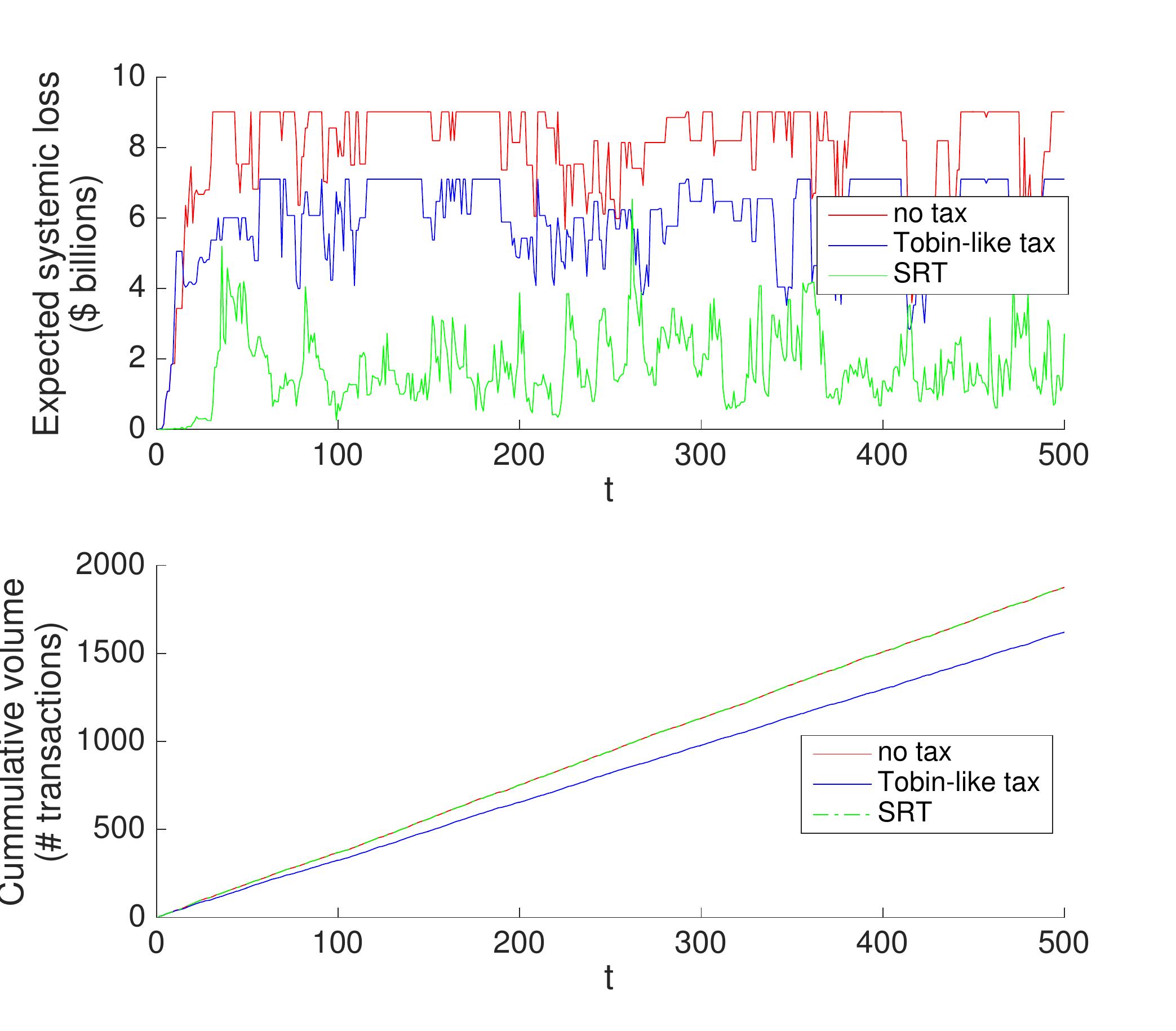}
  }
\caption{Evolution of Expected Systemic Loss: (i) without Tax (red); (ii) with a Tobin-like Tax (blue) and (iii) with a Systemic Risk Tax (SRT) (green). \textit{At each time $t$, the regulator finds the SRT $\hat{\mc{T}}$ by solving the optimization problem in Eqs. (\ref{eq:minProb})-(\ref{eq:minProbConstr}). The target volume $\nu$ is set to $Vol(\mu^*_t)$, the volume achieved by the un-taxed equilibrium matching. Top panel shows the expected systemic loss supposing a default event (i.e. conditioning on $t=T$). Bottom panel shows the cumulative transaction volume (the cumulative number of loans extended over time). Interbank system has $|\mc{N}|=10$ banks. There are $500$ time steps and each loan has a maturity of $S=30$ periods and a value of $1$ billion dollars. Model parameters are: $y=1$, $Z=0.5$, $Y^i_0 \sim U(0.5,2.5)$ (all values in billions of dollars), $r_i \sim U(0,8\%)$ and $\gamma^i \sim  U(0,0.09\%)$ are chosen randomly at time $t=0$. $\bar{r}_i = 9\%$. The Tobin-like tax is set to $\kappa = 3\%$. For simplicity, this simulation assumes that banks form the belief $\rho_{t,S}^j=\bar{\rho}_S^j$. This case is discussed in Section \ref{sec:lim_info}.}}
  \label{Fig1} 
\end{figure*}

\section{Conclusion}
\label{sec:conclusion}
In this article, we have shown analytically, without any recourse to computational or simulation techniques that a regulator (e.g. a Central Bank) possessing information about the topology of a financial network of assets and liabilities can design a \em transaction-specific \em tax that incentivizes institutions (e.g. banks) to create a network more resilient to insolvency cascades. This transaction-specific \em systemic risk tax \em (SRT), allows a regulator to select a unique equilibrium network configuration that minimizes systemic risk given a target transaction volume. Without this SRT, many equilibrium networks can arise and they are generally inefficient, i.e. they may present higher systemic risk. We also showed  analytically that a standard financial transaction tax (FTT) (e.g. a Tobin-like tax) reduces transaction volume while having only a marginal effect on reducing systemic risk. Indeed, a Tobin-like tax fails to account for the fact that different transactions have different impacts on creating systemic risk because they involve institutions of different systemic importance.

We note a number of differences between the current article and \cite{poledna2014elimination}, in which the concept of a SRT was introduced. 
In that paper, the validity of the SRT was shown with the help of an agent-based model (CRISIS macro-financial model), that is -- as in all such models -- based on a number of assumptions and parameters that may be hard to verify. The main contribution of this present work is showing  that the SRT works independently of the numerical model chosen, and that an equilibrium exists 
under this tax and is drastically different from untaxed equilibria in terms of the overall systemic risk levels. 
The analytical framework here allows us to prove several additional interesting properties of the SRT independently of numerical validation. 
Namely, the multiplicity of equilibrium matchings and the fact that the SRT can select a unique and efficient matching. It also provides the intuition that a Tobin tax reduces transaction volume because it prevents certain matchings from being formed, while the SRT has a nice interpretation as allowing the same volume to be exchanged under a different matching configuration.
Further, \cite{poledna2014elimination} covers the case where borrowers sequentially (one at a time) choose one out of many lenders. This can be seen as a special case of the current model, which examines multiple borrowers and multiple lenders simultaneously in a two-sided market.

While we illustrated these results with the help of a simple interbank network formation setup, it is important to emphasize that the concept of a SRT applies much more generally to any credit market. Indeed, we made minimal assumptions about the reasons that push different institutions to trade with one another. The concept of a SRT is based on the idea that the preferences of counter-parties for one another can be changed arbitrarily, thus leading to any desired equilibrium matching between counter-parties. This systemic risk tax was extensively simulated with an agent-based model (CRISIS macro-financial model, see \cite{poledna2014elimination}) and was shown to perform very well under a wide range of different conditions and parameters. The concept of a systemic risk tax was also applied to a more complex interbank system involving derivative contracts (\cite{leduc2016CDS}) and was shown to perform very well. It may be applied to other types of networked systems as well and this is left for future research. 

\newpage

\section{Appendix}
\label{sec:appendix}

\subsection{Limited Information}
\label{sec:lim_info}
So far we have assumed that at any time $t$, all information about the system was common knowledge. In other words, we assumed that the topology of the financial system, i.e. $\bar{A}_{t-1}$ and $\vec{E}_{t-1}$, were known to all banks. In reality, while $\vec{E}_{t-1}$ may be inferred from publicly available balance sheet information, the whole topology of the interbank system (i.e. $\bar{A}_{t-1}$) may not be available to all banks. There are two ways of dealing with this. One is to assume some common prior $\bar{q}$ on the contagion risk. The total probability of failure of some borrower $j$ would then be $\rho_{t,S}^j = \bar{\rho}_S^j + (1-\bar{\rho}_S^j)\bar{q}$. We may also assume that banks literally ignore contagion risk, in which case their belief about the total probability of failure of some borrower $j$ would simply be $\rho_{t,S}^j=\bar{\rho}_S^j$. This does not affect the results derived throughout the paper although in the second case we can predict higher trading volume. Indeed, the contagion risk not being taken into account when a lending bank sets the risk premium, this necessarily results in lower risk premia and thus higher demand for loans (fewer banks' reservation rates being exceeded).

\subsection{Risk Management Strategies with Strict Preferences for Borrowers}
\label{sec:other_risk_mang_strat}

In this section, we show that our results extend naturally to the case where lenders have \textit{strict} preferences over borrowers. Here we allow lenders to manage their risk by favoring borrowers with the lowest credit risk. We let the expected payoff of a lender $i \in \mc{L}_t$ that lends to a borrower $j \in \mc{B}_t$ be 
\begin{equation}
\Pi_{\lambda}^i(j) = \frac{1}{(1+r_i)^S} (1-\rho_{t,S}^j) (1 + r_i)^S-1.
\end{equation}
Here, the lender does not charge a risk premium $h_{ij}$ to hedge the credit risk of the borrower. This payoff is thus strictly decreasing in $\rho^j_{t,S}$, the borrower's default probability. It follows that lender $i$ would prefer lending to the safest bank, i.e. the bank $j$ with the lowest probability of default $\rho^j_{t,S}$ on a loan. Thus, when the $\rho^j_{t,S}$ can be strictly ordered, lender $i$ has a strictly ordered list of preferences, $P^i_{\lambda}$, on the set of potential borrowers $\mc{B}_t$. Lender $i$'s preferences are of the form $P^i_{\lambda}=d, e, f, ...$,  indicating that its first choice is to lend to borrower $d$, its second choice is to lend to borrower $e$, etc., where borrowers are ordered according to their probability of default, $\rho^d_{t,S} < \rho^e_{t,S} < \rho^f_{t,S} < ...$, on the loans they request.

The next result states that under such a regime, there  always exists a \textit{unique} equilibrium.

\begin{proposition}[Equilibrium Uniqueness with Strict Preferences]
Given any market for liquidity $(\mc{B}_t,\mc{L}_t, \textbf{P})$ where the preferences are strict, there exists a unique stable matching $\mu_t^*$ at time $t$. Moreover, the amount of liquidity exchanged at time $t$ is bounded as follows: $Vol(\mu_t^*) \leq min(|\mc{B}_t|,|\mc{L}_t|)$.

\label{prop:strict_pref_existence}
\end{proposition}
This equilibrium is generally systemic risk inefficient, as will be shown later in Proposition \ref{prop:ESL_strict_prefs}. This case is similar to the classical setting of \cite{GaleShapley1962}, where each side of the market has strictly ordered preferences. In our case, however, uniqueness follows from homogenous preferences on both sides of the market. 

With strict preferences on both sides of the market, a stable matching can naturally emerge from a process of repeated negotiations. Each borrower solicits lenders in decreasing order of their preferences, i.e. they first solicit the lender who offers the lowest interest rate and so on. The solicited lender gives provisory approval if the borrower is safer (i.e. has a smaller default probability) than the ones who have solicited him earlier. This process leads, in a finite number of iterations, to the unique stable matching $\mu_t^*$.

This uniqueness result also applies to the case of a Tobin-like tax $\kappa$, since the latter does not affect the ordering of preferences\footnote{It only makes certain borrowing rates too high for some borrowers (i.e. $r^{\kappa}_{i}>\bar{r}_j$, for some $i\in\mc{L}_t$ and some $j\in\mc{B}_t$).}. By an argument similar to that of Theorem \ref{th:SRT_unique_match}, it also follows that under an appropriately chosen SRT $\mc{T}$, any feasible matching such that $r_{i} < \bar{r}_j$ can be sustained as a unique equilibrium. 

The next proposition is analogous to Proposition \ref{th:ESL}, however for the strict preferences setting.

\begin{proposition}[Systemic Risk under Systemic Risk Tax with Strict Preferences]
Let $(\mc{B}_t,\mc{L}_t, \textbf{P})$ be a market for liquidity at time $t$. Given a net exposure matrix $\bar{A}_{t-1}$ at time $t-1$, let $\bar{A}^{*,\mc{T}}_{t}$, $\bar{A}^{*,\kappa}_{t}$ and $\bar{A}^*_{t}$ be the net exposure matrices formed at time $t$ with a systemic risk transaction tax matrix $\mc{T}$, with a Tobin-like tax $\kappa$, and without tax by the unique equilibrium matchings $\mu^{*,\mc{T}}_t$, $\mu^{*,\kappa}_t$ and $\mu^{*}_t$, respectively. Then,

\item (i) If $Vol(\mu^*_t) = \nu$, there exists $\mc{T}$ such that $ESL(\bar{A}^{*,\mc{T}}_{t},\vec{E}_{t}) \leq ESL(\bar{A}^*_{t},\vec{E}_{t})$ and $Vol(\mu^{*,\mc{T}}_t) \geq Vol(\mu^*_t)$; In particular, there exists $\mc{T}$  such that $\mu^{*,\mc{T}}_t$ is systemic risk efficient.
 
\item (ii) If $Vol(\mu^{*,\kappa}_t) = \nu$, there exists $\mc{T}$ such that $ESL(\bar{A}^{*,\mc{T}}_{t},\vec{E}_{t}) \leq ESL(\bar{A}^{*,k}_{t},\vec{E}_{t})$ and $Vol(\mu^{*,\mc{T}}_t) \geq Vol(\mu^{*,k}_t) $.
 
\label{prop:ESL_strict_prefs}
\end{proposition}

\subsection{Variable Loan Size}
\label{sec:VariableLoanSize}
We have assumed that a bank whose household clients demand a loan must obtain that money from the interbank market (from another bank whose household clients have just deposited money). The size of an interbank loan is thus exogenous and depends on the amount that a bank must borrow from the interbank market in order to fulfill the borrowing needs of the household clients. The amount of the loan is thus entirely governed by the exogenous household shock. In Section \ref{sec:liquidity_shocks}, we assumed this amount was $1$, allowing for a simple analysis of the matching process at every time time $t$. Since the network is formed of many overlapping matchings, the net exposure $\bar{A}_t^{ij}$ between two banks can have a variable (integer) size. The net exposure matrix $\bar{A}_t$ at any $t$ can thus be seen as a weighted graph.

\subsubsection{A More Complex Model with Variable Loan Sizes}
More complex matching models could be used to account for variable loan sizes. For example, \cite{baiou2002} consider a stable allocation problem, which can be used to study a set of employees each having a certain number of available hours to work and a set of employers each seeking a certain number of hours of work. Such an allocation problem could be applied to our context of lenders and borrowers who try to exchange different amounts of liquidity. The main moral hazard problem present in an interbank system would however remain: a bank does not consider the systemic risk imposed on other banks when it lends/borrows. Such more realistic matching models would mainly complicate the analysis and are thus left as future work.


\subsubsection{A Simple Extension Accounting for Variable Loan Sizes}

The current matching model could be extended to variable loan sizes in the following way:  

At time $t$, let each bank $i \in \mc{N}$ receive a liquidity shock $\tilde{\epsilon}_t^i = \epsilon_t^i \gamma_t^i$, where $\gamma_t^i$ is exponentially distributed and governs the magnitude of the liquidity shock, while $\epsilon_t^i$ is as in Section \ref{sec:liquidity_shocks} and governs whether a bank is a lender or a borrower for that period. Then the maximum amount of liquidity that can be exchanged between two trading partners $i$ and $j$ in a matching where lender $i$ lends to borrower $j$ is $\min(\tilde{\epsilon}^i_{t,k},|\tilde{\epsilon}^i_{t,k}|)$. 

Thus it is clear that a single round of matching between lenders and borrowers will not allow for the exchange of all the liquidity available. We can therefore allow for multiple matching ``rounds'' within a time period $t$. 

Then, for any round $k$, the supply and demand of each lender $i$ and borrower $j$ can be be updated as follows (assuming there is a match between $i$ and $j$ in round $k$) 

$$\tilde{\epsilon}^i_{t,k+1} = \tilde{\epsilon}^i_{t,k} -   \min(\tilde{\epsilon}^i_{t,k},|\tilde{\epsilon}^j_{t,k}|)$$

$$\tilde{\epsilon}^j_{t,k+1} = \tilde{\epsilon}^j_{t,k} +   \min(\tilde{\epsilon}^i_{t,k},|\tilde{\epsilon}^j_{t,k}|)$$

Since in any matching, one counter-party either lends all its current supply or borrows all its current demand, it is removed from the set of   lenders or borrowers for the next round.
It follows that $\mc{B}_{t,k+1} \subseteq \mc{B}_{t,k}$ and $\mc{L}_{t,k+1} \subseteq \mc{L}_{t,k} $, with one inclusion being strict when at least one stable matching exists in round $k$. Thus, we can match the updated sets of lenders and borrowers until the final round in which either one set becomes empty or no stable matching exists (e.g. because the remaining lenders offer rates that exceed the remaining borrowers' reservation rates). This process terminates in a finite number of rounds. Note  that in this multi-round model, a matching is only stable in a single round (agents are myopic and do not consider the subsequent rounds).

The result is a bipartite graph formed in period $t$, but this bipartite graph is different from the one in Fig. \ref{fig:matching} in two ways: (i) each borrower or lender can potentially have multiple edges incident on it; and (ii) the edges are weighted by the amounts of liquidity exchanged in each round (a real number).

This model however adds little to the analysis. Indeed, in the original model, the overlapping matchings created over time create net exposures of variable (integer) sizes between any two banks anyway. For example, the interbank network in Fig. \ref{fig:IB_network} has the net exposure matrix in Eq. (\ref{eq:A_bar_matrix}). In this matrix, bank $2$ has loaned twice to bank $1$ in the past and thus has a net exposure of size $2$. From the perspective of systemic risk, this can be seen as a loan of size $2$.

\subsection{Trading with Multiple Partners in a Single Period}
The network formation model is addressed as a bipartite matching at every time period. Indeed, household loans must be serviced by interbank loans and a bank can only use current household deposits to extend an interbank loan. This can be justified by assuming that a bank's liquidity shock $\epsilon_t^i$ is really the sum of all household deposits \emph{minus} the sum of all loans demanded by households. A bank then uses its household deposits to extend household loans. If the difference is positive, then the bank has a supply of liquidity ($\epsilon_t^i > 0$) that it can lend on the interbank market. If the difference is negative ($\epsilon_t^i<0$), then the bank must borrow on the interbank market to extend the remaining household loans. 

\subsubsection{More Complex Models}
In reality, however, the reasons that push banks to borrow/lend on the interbank market are more complex. Banks may borrow and lend in the same period. Under a good rate, a bank may also prefer to borrow than to lend, even if it has a supply of liquidity. To model such more realsitic interbank borrowing/lending decisions, more complex matching technologies such as the one studied in \cite{hatfield2013stability} or the multi-sided matching market studied in \cite{fleiner2016trading} could be used. Such matching models may be more realistic, but  might also obscure the analysis that we are interested in making, i.e. focusing on the effect of lending on systemic risk. In any event, the moral hazard problem would remain: a bank does not consider the systemic risk imposed on other banks when it lends/borrows. The idea of a SRT used to redress such inefficiencies could then also be studied. Using more general matching models to study the interbank market can thus be left as future work.

\subsubsection{Risk Diversification}

The basic model we study only allows a lending bank to offset counter-party  risk by setting a risk premium that is increasing in the default risk of the borrower. In Section \ref{sec:other_risk_mang_strat}, we also considered a different risk management strategy in which the lending bank has strict preferences and tries to lend to the safest borrower.

A more sophisticated risk management strategy that could be used by a bank is to diversify the credit risk of borrowers. One way to model this would be to use a matching process like that of \cite{alkan2003stable}, where each lender could have a choice function $C$ which, given the set of borrowers on the other side of the market, chooses the most preferred subset $S=C(\mc{B}_t)$ contained in $\mc{B}_t$. $S$ is then said to be \emph{revealed preferred} to all other subsets of $\mc{B}_t$. A lender could then choose to lend smaller amounts to a given number of borrowers, thereby diversifying its credit risk. This would lead to a denser network, although with smaller weights. On the other hand, for the same reasons as previously stated, inefficiencies and the moral hazard would still be present is such a model and it could be interesting to study how a SRT could be applied in this particular context, although it would operate according to the same idea. This may be explored in future research.

\subsection{Variable Time to Maturity}
\label{sec:VariableTimeToMaturity}
So far, we have assumed a fixed time to maturity $S$ for all loans. We can explain this again by saying that loan maturities (just like sizes) are governed by the needs of the borrowing households. Here we fix that to $S$. When a bank's household clients want to borrow more than they  deposit, the bank must use the interbank market to service this net household demand and secure a loan for $S$ periods. Since on the other side of the market, depositing households are assumed to also deposit money for $S$ periods, the two flows are matched in terms of maturity. This assumption simplifies the accounting. We fix $S$ mainly because we wish to abstract away from the  individual decisions that go into determining the maturity of a loan. This depends on things that we are not interested in modeling and which would obscure the analysis. In a more complex model using a more general matching technology, a lender could then have preferences over maturities (for example, preferring shorter maturities), but this is not explored in the current article. Moreover, it is not clear how what drives preferences for maturities would be directly relevant to studying systemic risk in this context. 

However, note that we could assign a random maturity to every loan (after the matching has been formed, presuming that maturity results from negotiations taking place after the two parties have agreed to trade). This would not change the equilibrium behavior, but would allow each loan to have a random (e.g. Poisson distributed) maturity. This would however add little to the analysis since at any given time, the interbank network created by the current model is a network of overlapping loans of various times-to-maturity (since the loans were created at different times in the past).

\subsection{Additional Numerical Results}

We now present additional network statistics for the interbank system simulated in Section \ref{sec:Reg_opt_prob}. 

\subsubsection{Degree Distribution}
Figure \ref{fig:in_out_deg} shows the empirical in- and out-degree distributions of the interbank network. The in-degree is the number of counter-parties from which a bank has borrowed money (irrespectively of the size of the exposure) and the out-degree is the number of counter-parties to which a bank has lent money (irrespectively of the size of the exposure)\footnote{This however takes the bilateral netting of exposures into account in the sense that if two banks have loaned an equivalent amount of money to each other, then the exposure is cancelled.}. We see that, although the number of counter-parties under a SRT appears to decrease on average, the distributions in the three scenarios are not strikingly different. The number of counter-parties is thus not a useful measure to understand how a SRT reshapes the interbank network.

\begin{figure*}[h]
  \centerline{
\includegraphics[scale=0.45]{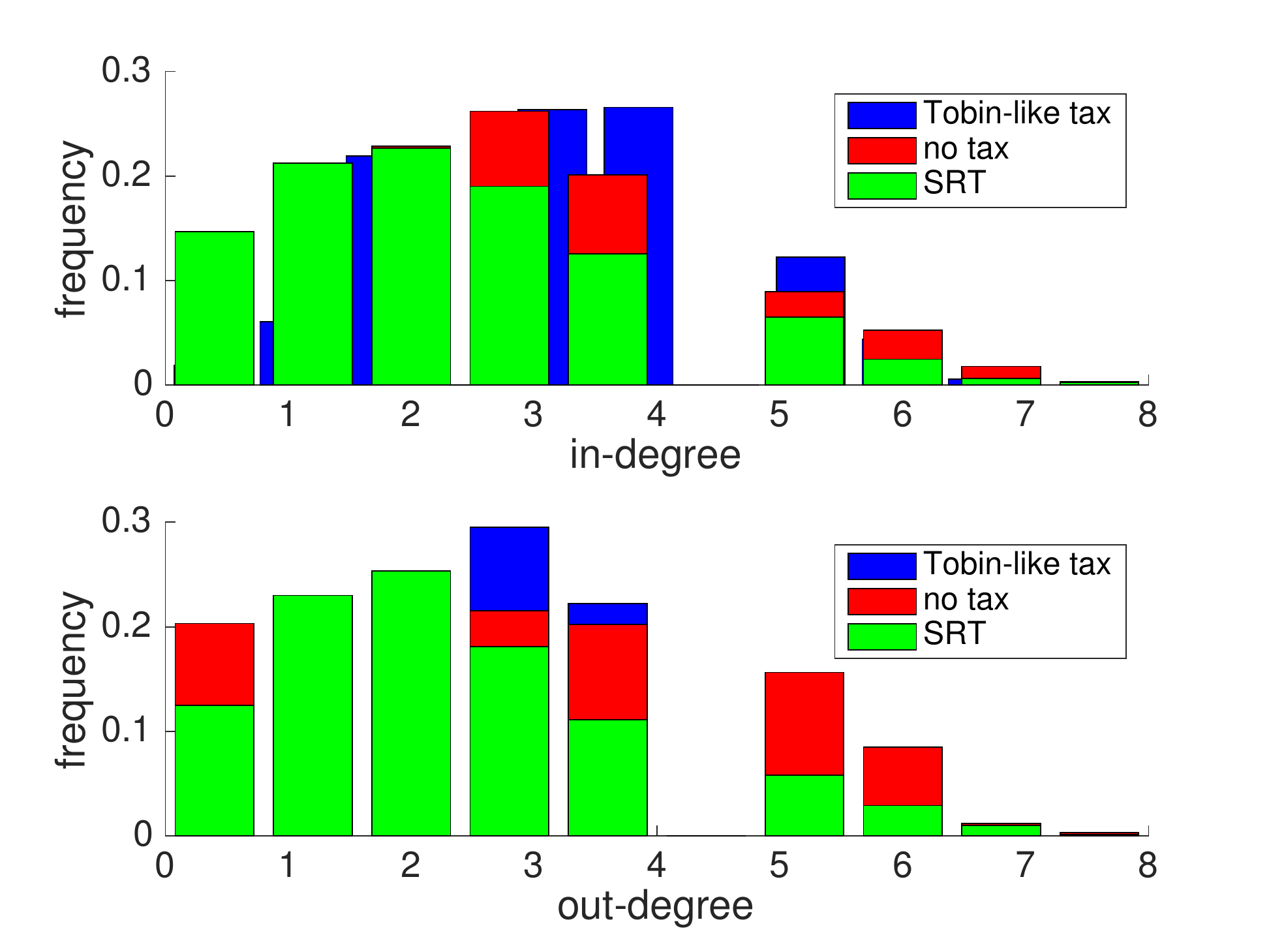}
  }
  \caption{Empirical distributions of in-degrees and out-degrees in the unweighted adjacency matrix of the interbank network simulated in Section \ref{sec:Reg_opt_prob}. \textit{The in-degree represents the number of counter-parties from which a bank has borrowed money. The out-degree represents the number of counter-parties to which a bank has lent money.}}
  \label{fig:in_out_deg} 
\end{figure*}

\subsubsection{Systemic Impact $SI^i$}
To understand what properties of the network change under the three different scenarios, we must look at different statistics. The most relevant statistic is $SI^i$, the systemic impact of a bank, as defined in Definition \ref{def:SI}. This is truly a centrality measure since it measures how many other banks are affected by the bankruptcy of a bank $i$. The empirical distribution of $SI^i$ is shown in Fig. \ref{fig:SI_emp_dist}. We clearly see that the SRT drastically shifts the distribution of systemic impacts towards lower values. A Tobin-like tax on the other hand, only has a marginal impact on it.

\begin{figure*}
  \centerline{
\includegraphics[scale=0.45]{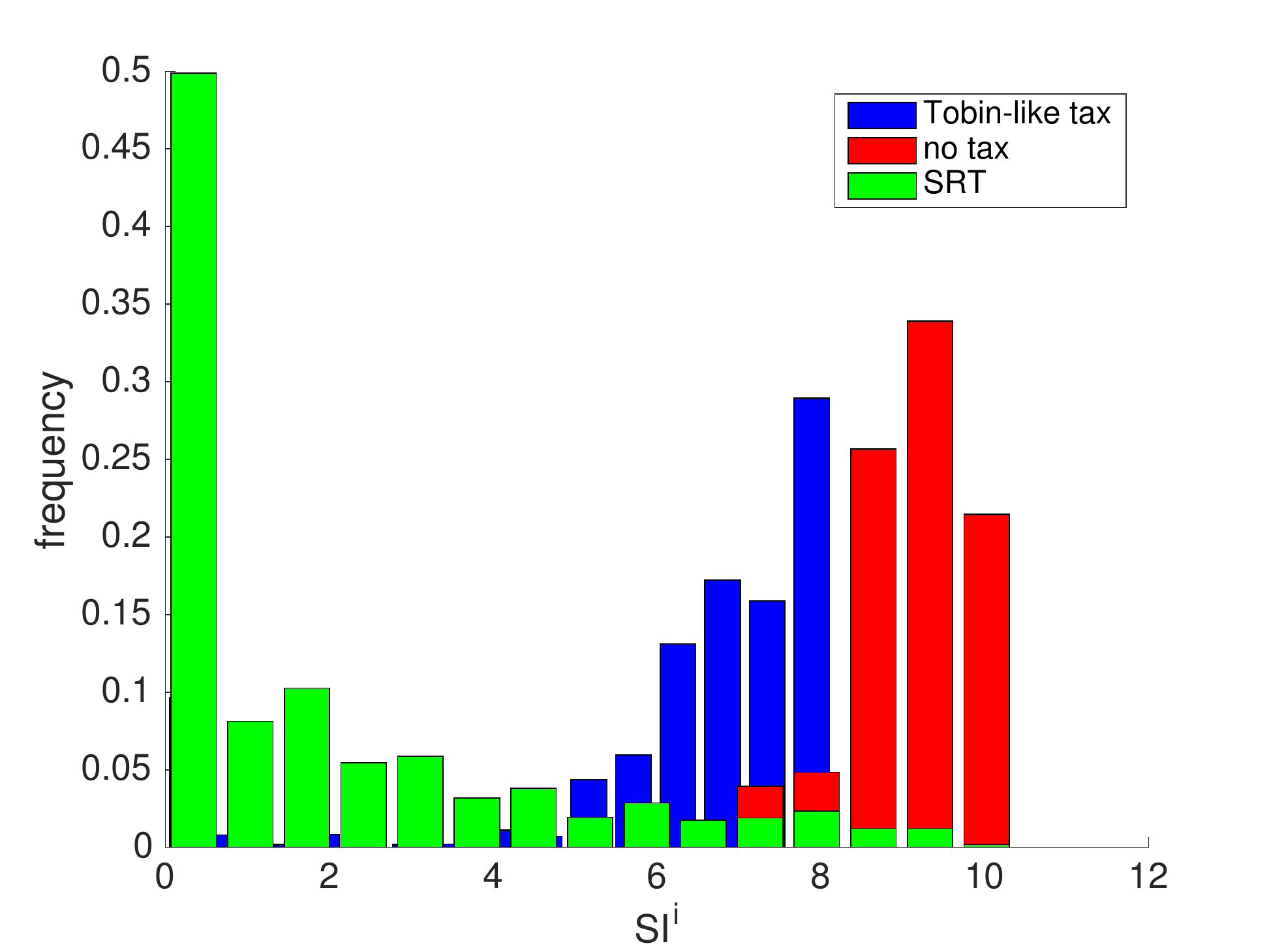}
  }
  \caption{Empirical distributions of the systemic impact $SI^i$ of a bank in the interbank network simulated in Section \ref{sec:Reg_opt_prob}.}
  \label{fig:SI_emp_dist} 
\end{figure*}

\subsubsection{Clustering Coefficient and Spectral Radius}
Other more standard network statistics may provide useful insights into how a SRT reshapes the interbank network. Figure \ref{fig:AvClustCoeff} shows the distribution of the average local clustering coefficient. The local clustering coefficient measures how close the neighbors of a bank are to being a clique (being all connected) and also indirectly provides information about the presence of cycles of exposures. Such cycles of exposures can create substantial levels of systemic risk because the insolvency of a bank may render other banks in the cycle insolvent (see for example, \cite{duffie2011}). The local clustering coefficient for a bank is then given by the proportion of links between the banks within its neighborhood divided by the number of links that could possibly exist between them. For the purpose of this calculation, we ignore the weights of the links (the size of the exposures). The average is then taken over all nodes in the network. Since we have $500$ time steps (and thus $500$ networks), we can compute a distribution of the average clustering coefficient. It is clear from Fig. \ref{fig:AvClustCoeff} that the SRT lowers the clustering coefficient of the interbank network, whereas the Tobin-like tax only marginally lowers it. Figures  \ref{fig:SI_emp_dist} and \ref{fig:AvClustCoeff} indicate that a SRT has the effect of cutting cycles of exposures. Such cycles of exposures can cause the bankruptcies of several banks following the default of one of them.

\begin{figure*}
  \centerline{
\includegraphics[scale=0.45]{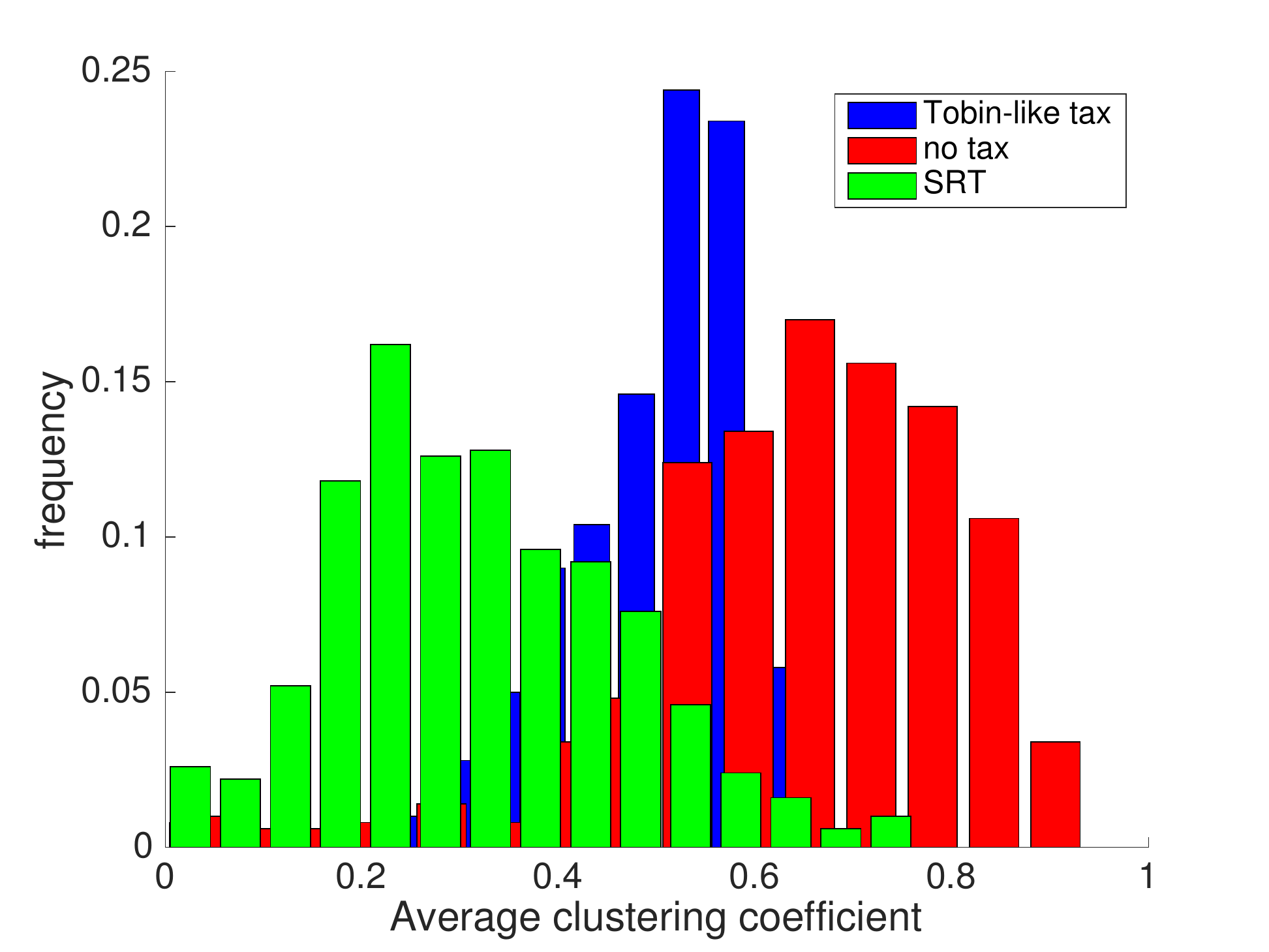}
  }
  \caption{Empirical distributions of the average clustering coefficient in the unweighted, undirected adjacency matrix of the interbank network simulated in Section \ref{sec:Reg_opt_prob}.}
  \label{fig:AvClustCoeff} 
\end{figure*}

In Fig. \ref{fig:SpectralRadius}, we look at the spectral radius. This is the magnitude of the largest eigenvalue of the unweighted, undirected adjacency matrix of the interbank network. Again we see how the SRT lowers the spectral radius whereas the Tobin-like tax has no noticeable impact on it.

\begin{figure*}
  \centerline{
\includegraphics[scale=0.45]{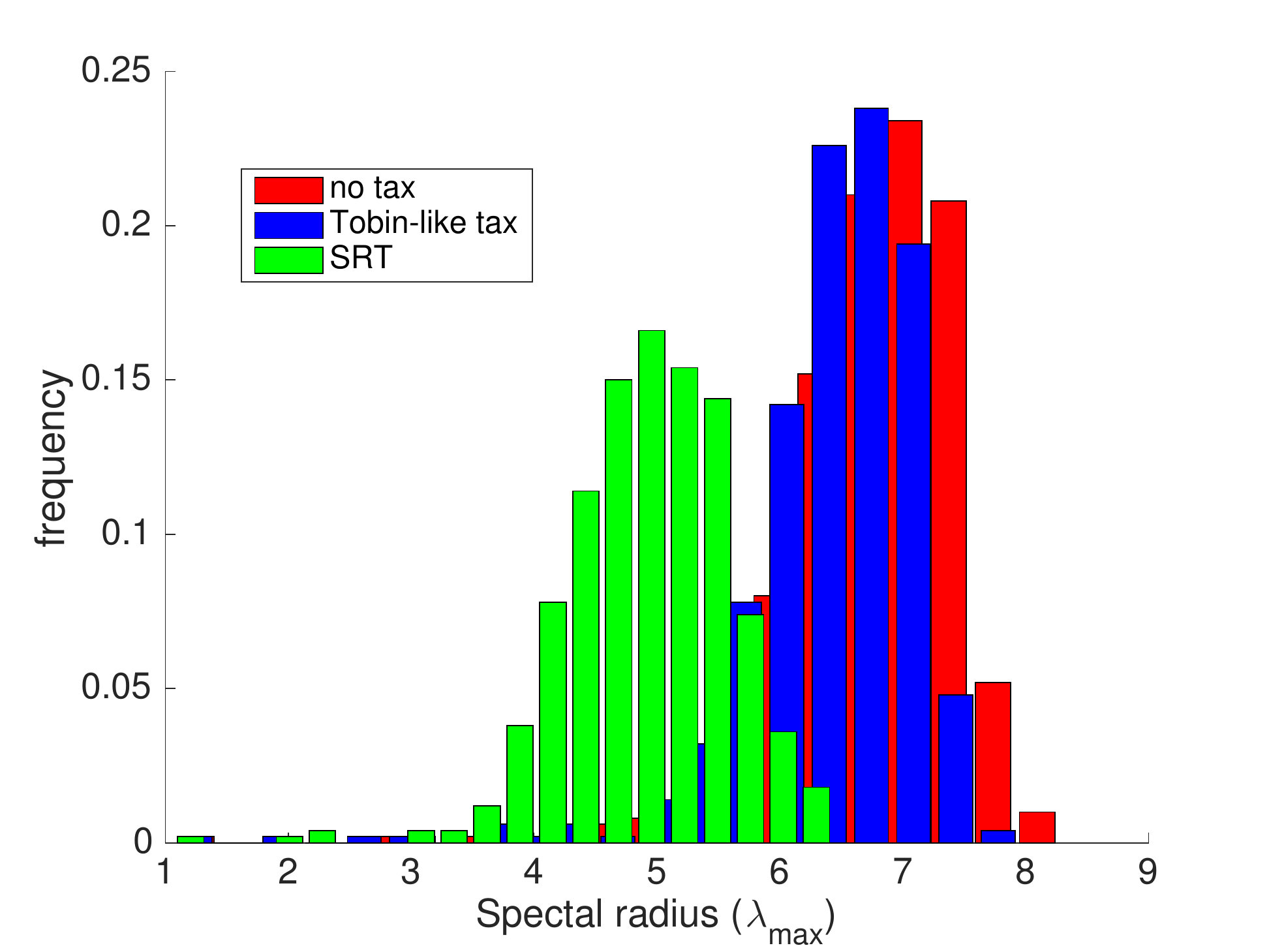}
  }
  \caption{Empirical distributions of the spectral radius in the unweighted, undirected adjacency matrix of the interbank network simulated in Section \ref{sec:Reg_opt_prob}.}
  \label{fig:SpectralRadius} 
\end{figure*}

\newpage

\subsection{Proofs}

\begin{proof}[Proof of Lemma \ref{lem:RiskPremia}]

A lender $i$'s expected payoff is
\begin{equation}
\Pi_{\lambda}^i(j) = \frac{1}{(1+r_i)^S} (1-\rho_{t,S}^j) (1 + r_i + h_{ij})^S-1.
\end{equation}

If $\rho_{t,S}^k=0$ (a risk-free loan), then the risk premium is $h_{ik}=0$. Lender $i$ will hedge its risk by charging a risk premium $h_{ij}$ such that the loan has an expected value as high as that of a risk-free loan. Then for any borrower $j$ with a default probability $\rho_{t,S}^j>0$, we have that $\Pi_{\lambda}^i(j) = \Pi_{\lambda}^i(k)$. Solving this yields $h_{ij} = \frac{1+r_i}{(1-\rho_{t,S}^j)^{1/S}}-1-r_i$.

\end{proof}

\begin{proof}[Proof of Proposition \ref{prop:EqMultiplicity}]
Let $j\in \mc{B}_t$ and $i \in \mc{L}_t$ and let $\mu_t$ be a matching in which $r_{ij} < \bar{r}_j$ for any $\mu_t(i)=j$ and $r_{i j} < r_{m j} $ for any $m \in \mc{L}_t$ such that $\mu_t(m)=m$. We will verify that it fulfills all the conditions of a stable matching (cf. Definition \ref{def:stable_matching}):

Condition (I) is trivially satisfied since lenders are indifferent $P_{\lambda}^i (k) \sim P_{\lambda}^i (j)$. The absence of strict preferences means they have no incentive to change the bank to which they lend. 

Condition (II) is satisfied since all borrowers have homogenous preferences and thus it cannot be that a given subset of borrowers $\vec{b} \in \mc{B}_t$ would agree to swap the lenders to which they are matched. Indeed for any $j \in \mc{B}_t$ and $i,l \in \mc{L}_t$, if $r_i < r_l$ then $r_{ij} < r_{lj}$ (since risk premia $h_{ij}$ and $h_{lj}$ do not affect the ordering of the lending banks in a borrower's preference list, they are determined only by the $r_i$'s).

Condition (III) is satisfied since we have assumed $\mu_t$ was a matching in which $r_{ij}<\bar{r}_j$ and thus $P_{\beta}^j (j) \prec P_{\beta}^j (\mu_t(j)) = P_{\beta}^j (i) $. Moreover, since we have assumed that $r_{ij}<r_{kj}$ for any $k \in \mc{L}_t$ such that $\mu_t(k)=k$, $P_{\beta}^j (k) \prec P_{\beta}^j (\mu_t(j)) $.

Hence $\mu_t$ is a stable matching and we may denote it by $\mu^*_t$.

The upper bound on $Vol(\mu_t^*)$ can be explained as follows:  Suppose all banks in the smaller set (either $\mc{B}_t$ or $\mc{L}_t$) are matched to a counter-party in the larger set, such that $r_{ij}<\bar{r}_j$. In such a case, $Vol(\mu_t^*) = min(|\mc{B}_t|,|\mc{L}_t|)$. On the other hand, if the smaller set is $\mc{B}_t$ and some borrowers are unable to find available lenders with  $r_{ij}<\bar{r}_j$, then they will remain unmatched, in which case $Vol(\mu_t^*) < min(|\mc{B}_t|,|\mc{L}_t|)$. Hence, in general $Vol(\mu_t^*) \leq min(|\mc{B}_t|,|\mc{L}_t|)$.

\end{proof}

\begin{proof}[Proof of Lemma \ref{lem:rho^i_S}]
Let $N_t^{agg} = \sum_{j \in \mc{N}} N_t^j$ be the sum of all counting processes and let $T$ be the first jump time of $N_t^{agg}$. It follows immediately that $T \sim exp(\gamma^{agg})$, where $\gamma^{agg} = \sum_{j \in \mc{N}} \gamma^j$. Let $t_i$ be the first jump time of $N_t^{i}$. Then the probability that bank $i$ is the first to fail exogenously in the $S$ periods ahead, $\bar{\rho}_S^i$, can be expressed as

\begin{eqnarray}
\bar{\rho}_S^i &=& \mathbbm{P}\{T=t_i,T\leq S\} \\
&=&  \mathbbm{P}\{T \leq S\} \cdot  \mathbbm{P}\{T=t_i  |T \leq S \} \\
&=& \int_{0}^{S} \gamma^{agg} e^{-\gamma^{agg} \cdot t} dt \cdot \frac{\gamma^i}{\gamma^{agg}} \\
&=& \big(1 - e^{-\gamma^{agg} \cdot S} \big) \frac{\gamma^i}{\gamma^{agg}}
\end{eqnarray} 

\end{proof}

\begin{proof}[Proof of Theorem \ref{th:SRT_unique_match}]

We now introduce a definition that will be useful to prove this result.

\begin{definition}
For each $j \in \mc{B}_t$, let $\mc{L}_t^j$ be the reduced set of lenders from whom $j$ is willing to borrow, i.e. $\mc{L}_t^j = \{ i: r_{ij} < \bar{r}_j \}$. We denote by $\overline{\mc{EQ}}_t$ be the set of feasible matchings (not necessarily stable) between the set of borrowers $\mc{B}_t$ and the reduced sets of lenders $\prod_{j \in \mc{B}_t} \mc{L}_t^j$, i.e. the set of one-to-one correspondences $\mu$ such that: (I) $\forall j \in \mc{B}_t$, $\mu(j)=i$, where $i \in \mc{L}_t^j$ (if $\mu(j) \neq j$). (II) $\forall i \in \mc{L}_t$, $\mu(i)=j$, where $j \in \mc{B}_t$ (if $\mu(i) \neq i$).
\label{def:M_t}
\end{definition}

In the above definition, $\overline{\mc{EQ}}_t$ is the set of feasible matchings, i.e. those in which a borrower $j$ pays a rate $r_{ij}$ below its reservation rate $\bar{r}_j$. Note that, by definition, $\overline{\mc{EQ}}_t \supseteq \mc{EQ}_t$.
We now show that any matching $\mu \in \overline{\mc{EQ}}_t$ can be made stable under an appropriate choice of transaction-specific tax $\mc{T}$.

For any $j \in \mc{B}_t$ and $i \in \mc{L}^j_t$, let $\tau_{ij} \in \mathbbm{R_+}$ be the tax imposed on a loan extended by  lender $i$ to borrower $j$. Since $r^{\mc{T}}_{ij} = r_i + h_{ij} + \tau_{ij}$, an appropriate choice of $\tau_{ij}$'s for all $i \in \mc{L}^j_t$ can reorder borrower $j$'s preference list $P_{\beta}^j$ arbitrarily.

Now given a desired matching $\hat{\mu}_t \in \overline{\mc{EQ}}_t$, we can construct preferences $\tilde{P}_{\beta}^j$ for all borrowers $j \in \mc{B}_t$ so that the resulting stable matching will be $\mu_t^*=\hat{\mu}_t$. To see this, let $\tau_{\hat{\mu}_t(j)  j}$ be such that $r^{\mc{T}}_{\hat{\mu}_t(j)  j} < r^{\mc{T}}_{kj}, \ \forall k \in \mc{L}^j_t$ s.t. $k \neq \hat{\mu}_t(j)$ and such that $r^{\mc{T}}_{\hat{\mu}_t(j)  j} < \bar{r}_j$ (if $j$ is to be matched) and $r^{\mc{T}}_{k,j} > \bar{r}_j \ \forall k \in \mc{L}^j_t$ if $\hat{\mu}_t(j)=j$ (if $j$ is to remain unmatched). Then, $\tilde{P}^j_{\beta} = \{ \hat{\mu}_t(j), \pi (\mc{L}^j_t \setminus \hat{\mu}_t(j)) \}$, where $\pi (\mc{L}^j_t \setminus \hat{\mu}_t(j))$ is some permutation of $j$'s reduced set of lenders $\mc{L}^j_t$ excluding the desired counter-party $\hat{\mu}_t(j)$. This results in preferences $ \{\tilde{P}^j_{\beta}\}_{j\in \mc{B}_t}$, such that the desired counter-party $\hat{\mu}_t(j)$ is on top of each borrower $j$'s preference list.

We now show that, under those tax-induced preferences $ \{\tilde{P}^j_{\beta}\}_{j\in \mc{B}_t}$, the matching $\hat{\mu}_t$ is stable.

Let all borrowers solicit the lender on top of their preference lists, then all lenders will accept to extend a loan since they are indifferent as to which borrowers they trade with. We will verify that each condition of Definition \ref{def:stableMatching} is satisfied.

Condition (I) is trivially satisfied since the lenders preferences are not strict, i.e. $P_{\lambda}^i (k) \sim P_{\lambda}^i (j)$ for all $i \in \mc{L}_t$ and $k,j \in \mc{B}_t$.

Condition (II) is satisfied since no borrower can improve his expected payoff by changing his counter-party. Doing so would force him to pay a higher rate $r_{ij}^{\mc{T}}$. So it follows that no group of borrowers $\vec{b} \subset \mc{B}_t$ would agree to swap counter-parties.

Condition (III) is satisfied since for any matched borrower $j$, $r^{\mc{T}}_{\hat{\mu}_t(j)  j} < \bar{r}_j$, and $r^{\mc{T}}_{\hat{\mu}_t(j)  j} < r^{\mc{T}}_{m j} $ for any other $m \in \mc{L}_t$.

We now prove uniqueness: 

Suppose there exists another stable matching $\mu^{'*}_t \neq \mu^*_t$. Then, by the construction of $\mc{T}$, a set $\vec{b} \subset \mc{B}_t$ of borrowers are not matched to the lenders on top of their respective taxed-induced preference lists. If the tax-induced preference of some borrower $j$ is to remain alone (i.e. $r^{\mc{T}}_{ij}>\bar{r}_j$), then condition (III) is violated and $\mu^{'*}_t$ cannot be stable. Otherwise, the members of $\vec{b}$ can agree to swap counter-parties so that they are matched with their top choices and thus condition (II) is violated. Thus $\mu^{'*}_t$ cannot be stable and we conclude that there exists a unique stable matching $\mu^{*}_t$, in which each borrower is matched with its (tax-induced) preferred lender or with itself (i.e. remains unmatched).

$\overline{\mc{EQ}}_t$ is thus the set of matchings that can be sustained as a unique equilibrium under an appropriate choice of the tax.

\end{proof}

\begin{proof}[Proof of Proposition \ref{prop:Tobin_tax}]

Let $\kappa > 0$ be a Tobin-like tax, i.e.

\begin{equation}
r^{\kappa}_{ij} = r_i + h_{ij} + \kappa
\end{equation}
for some lender $i\in\mc{L}_t$ and some borrower $j \in \mc{B}_t$. 

Part (i):  

This proof is identical to that or Proposition \ref{prop:EqMultiplicity}(i), but with $\mu_t$ being a matching in which $r^{\kappa}_{ij} < \bar{r}_j$ and $r^{\kappa}_{i j} < r^{\kappa}_{m j} $ for any $m \in \mc{L}_t$ such that $\mu_t(m)=m$.

Part (ii):  

Moreover, for each $j \in \mc{B}_t$, we can define a reduced set of lenders $\mc{L}_t^{j,\kappa} = \{i:r^{\kappa}_{ij} < \bar{r}_j\}$. It follows immediately that $\mc{L}_t^{j,\kappa} \subset \mc{L}_t^j$, where $\mc{L}_t^j = \{ i : r_{ij} < \bar{r}_j \}$. Preference lists over lenders, however, remain unchanged. Let $\bar{v}$ be the highest volume achievable by some matching in $\mc{EQ}_t$ and let $\mc{EQ}_t(\bar{v})$ be the set of equilibria achieving this volume. Then  $\mc{EQ}_t^{\kappa}(\bar{v}) \subseteq \emptyset \bigcup\mc{EQ}_t(\bar{v})$, noting that $\mc{EQ}_t^{\kappa}(\bar{v})$ may be empty. It thus follows that $\underset{\mu_t^{*,\kappa} \in \mc{EQ}_t^{\kappa}}{\text{max}} Vol(\mu_t^{*,\kappa}) \leq  \underset{\mu_t^{*} \in \mc{EQ}_t}{\text{max}} Vol(\mu_t^{*})$.

\end{proof}

\begin{proof}[Proof of Proposition \ref{th:ESL}]

\item Part (i):

Consider the set $\{ \mu_t \in \overline{\mc{EQ}}_t : Vol(\mu_t) \geq Vol(\mu^*_t) \}$, where $\overline{\mc{EQ}}_t$ is the set of feasible matchings defined in the proof of Theorem \ref{th:SRT_unique_match} (Definition \ref{def:M_t}). Since $\mu^*_t \in \overline{\mc{EQ}}_t$, this set is not empty. Let us now write

\begin{equation}
\hat{\mu}_t \in \underset{\{ \mu_t \in \overline{\mc{EQ}}_t : Vol(\mu_t) \geq Vol(\mu^*_t) \}} {argmin} ESL (A_{t}(\mu_t))
\label{eq:argmin}
\end{equation}
where $A_{t}(\mu_t)$ is the net exposure matrix formed at time $t$ with the matching $\mu_t$ (not necessarily stable).

There exists at least one such minimizer $\hat{\mu}_t$, since there can only be finitely many matchings $\mu_t \in \overline{\mc{EQ}}_t$ and $\{ \mu_t \in \overline{\mc{EQ}}_t : Vol(\mu_t) \geq Vol(\mu^*_t) \} \subseteq \overline{\mc{EQ}}_t$. 

Now from Theorem \ref{th:SRT_unique_match}, there exists $\mc{T}$ such that $\mu^{*,\mc{T}}_t = \hat{\mu}_t$. It follows that there exists $\mc{T}$ such that $ESL (A^{*,\mc{T}}_{t}) \leq ESL (A^*_{t})$ and $Vol(\mu^{*,\mc{T}}_t) \geq Vol(\mu^*_t)$.

\item Part (ii):

 Let $\kappa > 0$ be a Tobin-like tax, i.e.

\begin{equation}
r^{\kappa}_{ij} = r_i + h_{ij} + \kappa
\end{equation}
for some lender $i\in\mc{L}_t$ and some borrower $j \in \mc{B}_t$. As in the proof of Proposition \ref{prop:Tobin_tax} (ii), we can then define $\overline{\mc{EQ}}_t^{\kappa}$ as the set of possible matchings (not necessarily stable) under a Tobin-like tax $\kappa$ and  $\overline{\mc{EQ}}_t^{\kappa} \subseteq \overline{\mc{EQ}}_t$ since $\kappa$ is a particular case of $\mc{T}$. 

Let $\mu^{*,\kappa}_t \in \overline{\mc{EQ}}_t^{\kappa}$ be the equilibrium matching under the Tobin-like tax $\kappa$. Then $\mu^{*,\kappa}_t \in \overline{\mc{EQ}}_t^{\kappa} \subseteq \overline{\mc{EQ}}_t$. Since from Theorem \ref{th:SRT_unique_match}, we can design $\mc{T}$ such that $\hat{\mu}_t^{*,\mc{T}} = \mu_t$, for any $\mu_t \in \overline{\mc{EQ}}_t$, it then follows by an argument analogous to that of Part (i) that we can find a 
$\mc{T}$ such that $ESL (A^{*,\mc{T}}_{t}) \leq ESL (A^{*,\kappa}_{t})$ and $Vol(\mu^{*,\mc{T}}_t) \geq Vol(\mu^{*,\kappa}_t)$.

\end{proof}

\begin{proof}[Proof of Proposition \ref{prop:strict_pref_existence}]

We will first construct an equilibrium matching $\mu^*_t$ and then show that it is unique.

Each borrower $j\in \mc{B}_t$ has preferences strictly decreasing in the lenders' rates $r_i$, for $i\in \mc{L}_t$. Each lender $i\in \mc{L}_t$ has preferences strictly decreasing in the borrowers'  default probabilities $\rho^j_{t,S}$, for $j \in \mc{B}_t$. 

Let $i_1 \in \mc{L}_t$ denote the lender offering the lowest rate $r_{i_1}$. Let all borrowers (with $r_{i_1}<\bar{r}_j$) solicit this preferred lender $i_1$ for a loan. Then $i_1$ will respond by choosing the borrower $j_1\in \mc{B}_t$ with the lowest default probability $\rho^{j_1}_{t,S}$. This match of $j_1$ to $i_1$ is stable: it fulfills Condition (I) of Definition \ref{def:stableMatching}, since both $j_1$ and $i_1$ are matched to their preferred counter-party and would not want to deviate. It also satisfies Condition (III) since we assumed  $r_{i_1}<\bar{r}_{j_1}$. Condition II does not apply here since lenders' preferences are strict.

Let us now deal with the remaining unmatched members of the sets $\mc{B}_t$ and $\mc{L}_t$. Let $i_2 \in \mc{L}_t$ denote the lender offering the second lowest rate $r_{i_2}$. Let all unmatched borrowers (with $r_{i_2}<\bar{r}_j$) solicit this second preferred lender $i_2$ for a loan. Then $i_2$ will respond by choosing the borrower $j_2\in \mc{B}_t$ with the second lowest default probability $\rho^{j_2}_{t,S}$. As for the previous case, this match of $j_2$ to $i_2$ is stable, as it fulfills the conditions of Definition \ref{def:stableMatching}. Proceeding iteratively until $j_{N}$ and $i_{N}$, where $N=min(|\mc{B}_t|,|\mc{L}_t|)$ is the size of the smaller of the sets of borrowers and lenders, we have found a stable matching $\mu^*_t$ in which each bank is matched to its preferred available counter-party (or remains unmatched if there is no suitable available counter-party).

We will now show that $\mu^*_t$ is the unique stable matching. 

Let $\mu'_t$ be any matching in which $j_1$ is not matched to $i_1$. Then both $j_1$ and $i_1$ would benefit from giving up their currently-assigned counter-parties (under the matching $\mu'_t$) and trading together instead. Thus Condition (I) is violated and $\mu'_t$ is not stable.

Let $\mu''_t$ now be any matching in which $j_1$ is matched to $i_1$, but in which $j_2$ is not matched to $i_2$. Then again both $j_2$ and $i_2$ would benefit from giving up their currently-assigned counter-parties (under the matching $\mu''_t$) and trading together instead. Thus Condition (I) is violated and $\mu''_t$ is not stable. Proceeding iteratively until $j_{N}$ and $i_{N}$, where $N=min(|\mc{B}_t|,|\mc{L}_t|)$ is the size of the smaller of the sets of borrowers and lenders, we find at every step that Condition (I) is violated and thus that any matching different from $\mu^{*}_t$ is not stable. Since this covers all the possible matchings, we have shown that $\mu^*_t$ is the unique stable matching.

\end{proof}

\begin{proof}[Proof of Proposition \ref{prop:ESL_strict_prefs}]

We first prove that, as in Theorem \ref{th:SRT_unique_match}, any feasible matching can be sustained as the unique equilibrium under an appropriate choice of the SRT $\mc{T}$.

As in the proof of Theorem \ref{th:SRT_unique_match}, given any desired feasible matching $\mu_t \in \overline{\mc{EQ}}_t$, we can construct $\mc{T}$ so as to put the desired lender $i=\mu(j)$ on top of each borrower $j$'s preference list. Calling this matching $\mu^{*,\mc{T}}_t$, we now show that it is the unique stable matching. 

To show stability, we use the same intuition as in the proof of Proposition \ref{prop:strict_pref_existence}. Let $j_1\in \mc{B}_t$ be the borrower with the lowest default probability $\rho^{j_1}_{t,S}$. Then the match of borrower $j_1$ with lender $i_1=\mu^{*,\mc{T}}_t(j_1)$ is stable since both $j_1$ and $i_1$ are on top of each other's preference lists. We can now apply the same logic to show that the match of  $j_2 \in \mc{B}_t$, the borrower with the second lowest default probability $\rho^{j_2}_{t,S}$, to lender $i_2=\mu^{*,\mc{T}}_t(j_2)$ is stable. Proceeding iteratively until $j_{N}$ and $i_{N}$, where $N=min(|\mc{B}_t|,|\mc{L}_t|)$, we show that $\mu^{*,\mc{T}}_t$ is a stable matching in which each bank is matched to its preferred available counter-party under the SRT $\mc{T}$ (or remains unmatched if there is no suitable available counter-party).

Applying the exact same logic as for the uniqueness proof of Proposition \ref{prop:strict_pref_existence}, we conclude that $\mu^{*,\mc{T}}_t$ is the unique stable matching under the SRT $\mc{T}$.

We now prove Parts (i) and (ii) of the proposition.

Part (i):
The proof is identical to that of Proposition \ref{th:ESL} (i), since any feasible matching in $\overline{\mc{EQ}}_t$ can be sustained as a unique equilibrium by an appropriate choice of the SRT.

Part (ii):
The borrowing rate under a Tobin-like tax $\kappa$ here is simply defined as $r_{ij}^{\kappa} = r_i + \kappa$. The uniqueness result of Proposition \ref{prop:strict_pref_existence} holds under a Tobin-like tax $\kappa$ since it does not affect the preferences of the lenders, nor the ordering of the preferences of the borrowers. It only reduces the set of lenders with which a borrower would accept to trade (since $r_{ij}^{\kappa}$ may be greater than $\bar{r}_j$ for some lenders $i\in \mc{L}_t$).

The proof is then identical to that of Proposition \ref{th:ESL} (ii).

\end{proof}

\bibliographystyle{acmsmall}
\bibliography{SRT_Game-bibfile}

\begin{thebibliography}{}

\bibitem[\protect\citeauthoryear{Abdulkadiroglu and S{\"o}nmez}{Abdulkadiroglu
  and S{\"o}nmez}{2003}]{abdulkadiroglu2003school}
{\sc Abdulkadiroglu, A.} {\sc and} {\sc S{\"o}nmez, T.} 2003.
\newblock School choice: A mechanism design approach.
\newblock {\em The American Economic Review\/}~{\em 93,\/}~3, 729--747.

\bibitem[\protect\citeauthoryear{Acemoglu, Ozdaglar, and
  Tahbaz-Salehi}{Acemoglu et~al\mbox{.}}{2013}]{acemoglu2013systemic}
{\sc Acemoglu, D.}, {\sc Ozdaglar, A.}, {\sc and} {\sc Tahbaz-Salehi, A.} 2013.
\newblock Systemic risk and stability in financial networks.
\newblock Tech. rep., National Bureau of Economic Research.

\bibitem[\protect\citeauthoryear{Alkan and Gale}{Alkan and
  Gale}{2003}]{alkan2003stable}
{\sc Alkan, A.} {\sc and} {\sc Gale, D.} 2003.
\newblock Stable schedule matching under revealed preference.
\newblock {\em Journal of Economic Theory\/}~{\em 112,\/}~2, 289--306.

\bibitem[\protect\citeauthoryear{Amini, Cont, and Minca}{Amini
  et~al\mbox{.}}{2013}]{amini2013resilience}
{\sc Amini, H.}, {\sc Cont, R.}, {\sc and} {\sc Minca, A.} 2013.
\newblock Resilience to contagion in financial networks.
\newblock {\em Mathematical finance\/}.

\bibitem[\protect\citeauthoryear{Anufriev, Deghi, Panchenko, and
  Pinotti}{Anufriev et~al\mbox{.}}{2016}]{anufriev2016model}
{\sc Anufriev, M.}, {\sc Deghi, A.}, {\sc Panchenko, V.}, {\sc and} {\sc
  Pinotti, P.} 2016.
\newblock A model of network formation for the overnight interbank market.
\newblock {\em CIFR Paper\/}~103.

\bibitem[\protect\citeauthoryear{Babus}{Babus}{2016}]{babus2016formation}
{\sc Babus, A.} 2016.
\newblock The formation of financial networks.
\newblock {\em The RAND Journal of Economics\/}~{\em 47,\/}~2, 239--272.

\bibitem[\protect\citeauthoryear{Ba{\"\i}ou and Balinski}{Ba{\"\i}ou and
  Balinski}{2002}]{baiou2002}
{\sc Ba{\"\i}ou, M.} {\sc and} {\sc Balinski, M.} 2002.
\newblock The stable allocation (or ordinal transportation) problem.
\newblock {\em Mathematics of Operations Research\/}~{\em 27,\/}~3, 485--503.

\bibitem[\protect\citeauthoryear{Battiston, Puliga, Kaushik, Tasca, and
  Caldarelli}{Battiston et~al\mbox{.}}{2012}]{battiston2012debtrank}
{\sc Battiston, S.}, {\sc Puliga, M.}, {\sc Kaushik, R.}, {\sc Tasca, P.}, {\sc
  and} {\sc Caldarelli, G.} 2012.
\newblock Debtrank: Too central to fail? financial networks, the fed and
  systemic risk.
\newblock {\em Scientific reports\/}~{\em 2}.

\bibitem[\protect\citeauthoryear{Boss, Elsinger, Summer, and Thurner~4}{Boss
  et~al\mbox{.}}{2004}]{boss2004network}
{\sc Boss, M.}, {\sc Elsinger, H.}, {\sc Summer, M.}, {\sc and} {\sc Thurner~4,
  S.} 2004.
\newblock Network topology of the interbank market.
\newblock {\em Quantitative Finance\/}~{\em 4,\/}~6, 677--684.

\bibitem[\protect\citeauthoryear{Duffie and Zhu}{Duffie and
  Zhu}{2011}]{duffie2011}
{\sc Duffie, D.} {\sc and} {\sc Zhu, H.} 2011.
\newblock Does a central clearing counterparty reduce counterparty risk?
\newblock {\em Review of Asset Pricing Studies\/}~{\em 1,\/}~1, 74--95.

\bibitem[\protect\citeauthoryear{Eisenberg and Noe}{Eisenberg and
  Noe}{2001}]{eisenberg2001systemic}
{\sc Eisenberg, L.} {\sc and} {\sc Noe, T.~H.} 2001.
\newblock Systemic risk in financial systems.
\newblock {\em Management Science\/}~{\em 47,\/}~2, 236--249.

\bibitem[\protect\citeauthoryear{Elliott, Golub, and Jackson}{Elliott
  et~al\mbox{.}}{2015}]{ElliotGolubJackson2014}
{\sc Elliott, M.}, {\sc Golub, B.}, {\sc and} {\sc Jackson, M.~O.} 2015.
\newblock Financial networks and contagion.
\newblock {\em American Economic Review\/}.

\bibitem[\protect\citeauthoryear{Farboodi}{Farboodi}{2014}]{farboodi2014intermediation}
{\sc Farboodi, M.} 2014.
\newblock Intermediation and voluntary exposure to counterparty risk.
\newblock {\em Available at SSRN 2535900\/}.

\bibitem[\protect\citeauthoryear{Fleiner, Jank{\'o}, Tamura, and
  Teytelboym}{Fleiner et~al\mbox{.}}{2016}]{fleiner2016trading}
{\sc Fleiner, T.}, {\sc Jank{\'o}, Z.}, {\sc Tamura, A.}, {\sc and} {\sc
  Teytelboym, A.} 2016.
\newblock Trading networks with bilateral contracts.

\bibitem[\protect\citeauthoryear{Furfine}{Furfine}{2003}]{furfine2003interbank}
{\sc Furfine, C.} 2003.
\newblock Interbank exposures: Quantifying the risk of contagion.
\newblock {\em Journal of Money, Credit, and Banking\/}~{\em 35,\/}~1,
  111--128.

\bibitem[\protect\citeauthoryear{Gai and Kapadia}{Gai and
  Kapadia}{2010}]{gai2010contagion}
{\sc Gai, P.} {\sc and} {\sc Kapadia, S.} 2010.
\newblock Contagion in financial networks.
\newblock In {\em Proceedings of the Royal Society of London A: Mathematical,
  Physical and Engineering Sciences}. The Royal Society, rspa20090410.

\bibitem[\protect\citeauthoryear{Gale and Shapley}{Gale and
  Shapley}{1962}]{GaleShapley1962}
{\sc Gale, D.} {\sc and} {\sc Shapley, L.~S.} 1962.
\newblock College admissions and the stability of marriage.
\newblock {\em American mathematical monthly\/}, 9--15.

\bibitem[\protect\citeauthoryear{Georg}{Georg}{2011}]{georg2011basel}
{\sc Georg, C.-P.} 2011.
\newblock Basel iii and systemic risk regulation-what way forward?
\newblock Tech. rep., Working Papers on Global Financial Markets.

\bibitem[\protect\citeauthoryear{Glasserman and Young}{Glasserman and
  Young}{2015}]{glasserman2015likely}
{\sc Glasserman, P.} {\sc and} {\sc Young, H.~P.} 2015.
\newblock How likely is contagion in financial networks?
\newblock {\em Journal of Banking \& Finance\/}~{\em 50}, 383--399.

\bibitem[\protect\citeauthoryear{Hatfield, Kominers, Nichifor, Ostrovsky, and
  Westkamp}{Hatfield et~al\mbox{.}}{2013}]{hatfield2013stability}
{\sc Hatfield, J.~W.}, {\sc Kominers, S.~D.}, {\sc Nichifor, A.}, {\sc
  Ostrovsky, M.}, {\sc and} {\sc Westkamp, A.} 2013.
\newblock Stability and competitive equilibrium in trading networks.
\newblock {\em Journal of Political Economy\/}~{\em 121,\/}~5, 966--1005.

\bibitem[\protect\citeauthoryear{Irving}{Irving}{1994}]{irving1994stable}
{\sc Irving, R.~W.} 1994.
\newblock Stable marriage and indifference.
\newblock {\em Discrete Applied Mathematics\/}~{\em 48,\/}~3, 261--272.

\bibitem[\protect\citeauthoryear{Jackson and Watts}{Jackson and
  Watts}{2002}]{JacksonWatts2002}
{\sc Jackson, M.~O.} {\sc and} {\sc Watts, A.} 2002.
\newblock The evolution of social and economic networks.
\newblock {\em Journal of Economic Theory\/}~{\em 106,\/}~2, 265--295.

\bibitem[\protect\citeauthoryear{Jackson and Wolinsky}{Jackson and
  Wolinsky}{1996}]{jackson1996strategic}
{\sc Jackson, M.~O.} {\sc and} {\sc Wolinsky, A.} 1996.
\newblock A strategic model of social and economic networks.
\newblock {\em Journal of economic theory\/}~{\em 71,\/}~1, 44--74.

\bibitem[\protect\citeauthoryear{Leduc, Poledna, and Thurner}{Leduc
  et~al\mbox{.}}{2016}]{leduc2016CDS}
{\sc Leduc, M.~V.}, {\sc Poledna, S.}, {\sc and} {\sc Thurner, S.} 2016.
\newblock Systemic risk management in financial networks with credit default
  swaps.
\newblock {\em Available at http://ssrn.com/abstract=2713200\/}.

\bibitem[\protect\citeauthoryear{Matheson}{Matheson}{2012}]{matheson2012security}
{\sc Matheson, T.} 2012.
\newblock Security transaction taxes: issues and evidence.
\newblock {\em International Tax and Public Finance\/}~{\em 19,\/}~6, 884--912.

\bibitem[\protect\citeauthoryear{McCulloch and Pacillo}{McCulloch and
  Pacillo}{2011}]{mcculloch2011tobin}
{\sc McCulloch, N.} {\sc and} {\sc Pacillo, G.} 2011.
\newblock The tobin tax: a review of the evidence.
\newblock {\em IDS Research Reports\/}~{\em 2011,\/}~68, 1--77.

\bibitem[\protect\citeauthoryear{Poledna, Bochmann, and Thurner}{Poledna
  et~al\mbox{.}}{2017}]{poledna2016basel}
{\sc Poledna, S.}, {\sc Bochmann, O.}, {\sc and} {\sc Thurner, S.} 2017.
\newblock Basel iii capital surcharges for g-sibs fail to control systemic risk
  and can cause pro-cyclical side effects.
\newblock {\em Journal of Economic Dynamics and Control 77 (2017)
  230-246\/}~{\em 77}, 230--246.

\bibitem[\protect\citeauthoryear{Poledna, Molina-Borboa, van~der Leij,
  Mart\'{i}nez-Jaramillo, and Thurner}{Poledna
  et~al\mbox{.}}{2015}]{polednaMEXICO}
{\sc Poledna, S.}, {\sc Molina-Borboa, J.~L.}, {\sc van~der Leij, M.}, {\sc
  Mart\'{i}nez-Jaramillo, S.}, {\sc and} {\sc Thurner, S.} 2015.
\newblock Multi-layer network nature of systemic risk in financial networks and
  its implications.
\newblock {\em Journal of Financial Stability\/}~{\em 20}, 70--81.

\bibitem[\protect\citeauthoryear{Poledna and Thurner}{Poledna and
  Thurner}{2016}]{poledna2014elimination}
{\sc Poledna, S.} {\sc and} {\sc Thurner, S.} 2016.
\newblock Elimination of systemic risk in financial networks by means of a
  systemic risk transaction tax.
\newblock {\em Quantitative Finance\/}~{\em 16,\/}~10, 1599--1613.

\bibitem[\protect\citeauthoryear{Roth}{Roth}{1984}]{roth1984evolution}
{\sc Roth, A.~E.} 1984.
\newblock The evolution of the labor market for medical interns and residents:
  a case study in game theory.
\newblock {\em The Journal of Political Economy\/}, 991--1016.

\bibitem[\protect\citeauthoryear{Roth and Peranson}{Roth and
  Peranson}{1999}]{roth1999redesign}
{\sc Roth, A.~E.} {\sc and} {\sc Peranson, E.} 1999.
\newblock The redesign of the matching market for american physicians: Some
  engineering aspects of economic design.
\newblock Tech. rep., National bureau of economic research.

\bibitem[\protect\citeauthoryear{Roth, Sonmez, and Unver}{Roth
  et~al\mbox{.}}{2003}]{roth2003kidney}
{\sc Roth, A.~E.}, {\sc Sonmez, T.}, {\sc and} {\sc Unver, M.~U.} 2003.
\newblock Kidney exchange.
\newblock Tech. rep., National Bureau of Economic Research.

\bibitem[\protect\citeauthoryear{Roth and Sotomayor}{Roth and
  Sotomayor}{1992}]{roth1992two}
{\sc Roth, A.~E.} {\sc and} {\sc Sotomayor, M. A.~O.} 1992.
\newblock Two-sided matching: A study in game-theoretic modeling and analysis.

\bibitem[\protect\citeauthoryear{Summers and Summers}{Summers and
  Summers}{1989}]{summers1989financial}
{\sc Summers, L.~H.} {\sc and} {\sc Summers, V.~P.} 1989.
\newblock When financial markets work too well: a cautious case for a
  securities transactions tax.
\newblock {\em Journal of financial services research\/}~{\em 3,\/}~2-3,
  261--286.

\bibitem[\protect\citeauthoryear{Thurner and Poledna}{Thurner and
  Poledna}{2013}]{thurner2013debtrank}
{\sc Thurner, S.} {\sc and} {\sc Poledna, S.} 2013.
\newblock Debtrank-transparency: Controlling systemic risk in financial
  networks.
\newblock {\em Scientific reports\/}~{\em 3}.

\bibitem[\protect\citeauthoryear{Tobin}{Tobin}{1978}]{tobin1978proposal}
{\sc Tobin, J.} 1978.
\newblock A proposal for international monetary reform.
\newblock {\em Eastern economic journal\/}~{\em 4,\/}~3/4, 153--159.

\bibitem[\protect\citeauthoryear{Zawadowski}{Zawadowski}{2013}]{zawadowski2013entangled}
{\sc Zawadowski, A.} 2013.
\newblock Entangled financial systems.
\newblock {\em Review of Financial Studies\/}~{\em 26,\/}~5, 1291--1323.

\end{thebibliography}

\end{document}